\documentclass[preprint,authoryear,5p]{elsarticle}
\usepackage{amssymb}
\usepackage{booktabs}
\usepackage{amsmath}
\usepackage{amsthm}
\usepackage[linesnumbered,longend,ruled,vlined]{algorithm2e}
\usepackage{pgfplots}
\usepackage{paralist}
\usepackage{url}
\usetikzlibrary{positioning,shapes,shadows,arrows}
\newtheorem{theorem}{Theorem}

\newtheorem{remark}{Remark}

\title{Heuristic algorithms for the bipartite unconstrained 0-1 quadratic programming problem\tnoteref{nserc}}

\tnotetext[nserc]{This research work was supported by an NSERC Discovery accelerator supplement awarded to Abraham P. Punnen.}
\tnotetext[nserc]{This research was completed during the visit of Daniel Karapetyan to the Simon Fraser University supported by LANCS International Scientific Outreach Fund.}

\author[nott]{Daniel Karapetyan\corref{cor1}}
\ead{daniel.karapetyan@gmail.com}
\author[sfu]{Abraham P.~Punnen}
\ead{apunnen@sfu.ca}

\cortext[cor1]{Corresponding author.}
\address[nott]{ASAP Research Group, School of Computer Science, University of Nottingham, Jubilee Campus, Wollaton Road, Nottingham, NG8\,1BB, UK}
\address[sfu]{Department of Mathematics, Simon Fraser University Surrey, Central City, 250-13450 102nd AV, Surrey, British Columbia, V3T\,0A3, Canada}

\renewcommand{\comma}{,\allowbreak\ }
\newcommand{\suchthat}{:\allowbreak\,}
\newcommand{\connect}[1]{\allowbreak\ \text{#1}\ \allowbreak{}}
\newcommand{\connectcomma}[1]{,\allowbreak\ \text{#1}\ \allowbreak{}}

\newcommand{\zerovector}[1]{\ensuremath{\mathbf{0}^{#1}}}
\newcommand{\partition}[0]{\ensuremath{\mathfrak{I}}}

{\par\bigskip\noindent\ignorespaces\begin{tabular*}{\textwidth}{@{} p{0.35\textwidth} @{\hspace{0.5em}} l} #1}%
{\end{tabular*}\\\smallskip\par\noindent\ignorespacesafterend}

\date{}

\begin{document}

\begin{abstract}
We study the Bipartite Unconstrained 0-1 Quadratic Programming Problem (BQP) which is a relaxation of the Unconstrained 0-1 Quadratic Programming Problem (QP).  Applications of the BQP include mining discrete patterns from binary data, approximating matrices by rank-one binary matrices, computing cut-norm of a matrix, and solving optimization problems such as maximum weight biclique, bipartite maximum weight cut, maximum weight induced subgraph of a bipartite graph, etc.  We propose several classes of heuristic approaches to solve the BQP and discuss a number of construction algorithms, local search algorithms and their combinations.  Results of extensive computational experiments are reported to establish the practical performance of our algorithms.  For this purpose, we propose several sets of test instances based on various applications of the BQP. Our algorithms are compared with state-of-the-art heuristics for QP which can also be used to solve BQP with reformulation. We also study theoretical properties of the neighborhoods and algorithms. In particular, we establish complexity of all neighborhood search algorithms and establish tight worst-case performance ratio for the greedy algorithm.

\end{abstract}

\maketitle

\section{Introduction}
\label{sec:introduction}

The Unconstrained 0-1 Quadratic Programming Problem (QP) is to
\begin{align*}
\text{maximize } & f(x) = x^T Q' x + c' x + c'_0\\
\mbox{subject to } & x \in \{ 0, 1 \}^n,
\end{align*}
where $Q'$ is an $n \times n$ real matrix, $c'$ is a row vector in $\mathbb{R}^n$, and $c'_0$ is a constant.  QP is a well-studied problem in the operations research literature~\citep{Billionnet2004}.  The focus of this paper is on a problem closely related to the QP called the \emph{Bipartite Unconstrained 0-1 Quadratic Programming Problem} (BQP)~\citep{Punnen2012}\@.  The BQP can be defined as follows:
\begin{align*}
\text{maximize } & f(x, y) = x^TQy + cx + dy + c_0\\
\text{subject to } & x \in \{ 0, 1 \}^m, y \in \{ 0, 1 \}^n,
\end{align*}
where $Q = (q_{ij})$ is an $m \times n$ real matrix, $c = (c_1, c_2, \ldots, c_m)$ is a row vector in $\mathbb{R}^m$, $d = (d_1, d_2, \ldots, d_n)$ is a row vector in $\mathbb{R}^n$, and $c_0$ is a constant.  Without loss of generality, we assume that $m \leq n$ and $c_0 = 0$.  In what follows, we denote a BQP instance built on matrix $Q$, row vectors $c$ and $d$ and $c_0 = 0$ as BQP$(Q, c, d)$, and $(x, y)$ is a feasible solution of the BQP if $x \in \{ 0, 1 \}^m$ and $y \in \{ 0, 1 \}^n$.  Also $x_i$ stands for the $i$th component of the vector $x$ and $y_j$ stands for the $j$th component of the vector $y$.

By a simple transformation, the BQP can be formulated as a QP of size $m + n$, see~\citep{Punnen2012}.  Since any feasible solution of such a QP instance corresponds to a feasible solution of the original BQP, both exact and heuristic algorithms available to solve the QP can be used directly to solve the BQP\@.  However, solving BQP instances by converting them into QP instances and then applying QP solvers is rather inefficient; indeed, the obtained QP instances are of larger size, and,  further, QP algorithms cannot exploit the special structure of the problem. Later in this paper, we discuss this assertion in detail, supported by experimental results.  Thus, in this paper we focus on heuristics designed specifically for solving BQP exploiting its special structure.

A graph theoretic interpretation of the BQP can be given as follows~\citep{Punnen2012}.  Let $I = \{ 1, 2, \ldots, m \}$ and $J = \{ 1, 2, \ldots, n \}$.  Consider a bipartite graph $G = (I, J, E)$.  For each node $i \in I$ and $j \in J$, respective costs $c_i$ and $d_j$ are prescribed.  Further, for each $(i,j) \in E$, a cost $q_{ij}$ is given.  Then the \emph{Maximum Weight Induced Subgraph Problem} on $G$ is to find a subgraph $G' = (I', J', E')$ such that $\sum_{i \in I'} c_i + \sum_{j \in J'} d_j + \sum_{(i,j) \in E'} q_{ij}$ is maximized, where $I' \subseteq I$, $J' \subseteq J$ and $G'$ is induced by $I' \cup J'$.  The Maximum Weight Induced Subgraph Problem on $G$ is precisely the BQP, where $q_{ij} = 0$ if $(i, j) \notin E$.

There are some other well known combinatorial optimization problems that can be modelled as BQP\@.  Consider the bipartite graph $G = (I, J, E)$ with $w_{ij}$ being the weight of the edge $(i, j) \in E$\@.  Then the \emph{Maximum Weight Biclique Problem} (MWBP)~\citep{Ambuhl2011,Tan2008} is to find a biclique in $G$ of maximum total edge-weight.  Define
\[
q_{ij} =
\begin{cases}
w_{ij} &\mbox{if } (i,j) \in E, \\
-M & \mbox{otherwise,}
\end{cases}
\]
where $M$ is a large positive number.  Set $c$ and $d$ as zero vectors.  Then BQP$(Q, c, d)$ solves the MWBP~\citep{Punnen2012}. This immediately shows that the BQP is NP-hard and one can also establish some approximation hardness results with appropriate assumptions~\citep{Ambuhl2011,Tan2008}.  The MWBP has applications in data mining, clustering and bioinformatics~\citep{Chang2012,Tanay2002} which in turn become applications of BQP.

Another application of BQP arises in approximating a matrix by a rank-one binary matrix~\citep{Gillis2011,Koyuturk2005,Koyuturk2006,Lu2011,Shen2009}. For example, let $H = (h_{ij})$ be a given $m \times n$ matrix and we want to find an $m \times n$ matrix $A = (a_{ij})$, where $a_{ij} = u_i v_j$ and $u_i, v_j \in \{ 0, 1 \}$, such that $\sum_{i = 1}^m \sum_{j = 1}^n (h_{ij} - u_i v_j)^2$ is minimized. The matrix $A$ is called a rank one approximation of $H$ and can be identified by solving the BQP with $q_{ij} = 1 - 2h_{ij}$, $c_i = 0$ and $d_j = 0$ for all $i \in I$ and $j \in J$.  Binary matrix factorization is an important topic in mining discrete patterns in binary data~\citep{Lu2011,Shen2009}. If $u_i$ and $v_j$ are required to be in $ \{-1,1\}$ then also the resulting factorization problem can be formulated as a BQP.

The Maximum Cut Problem on a bipartite graph (MaxCut) can be formulated as a BQP~\citep{Punnen2012} and this gives yet another application of the model.  BQP can also be used to find approximations to the cut-norm of a matrix~\citep{Alon2006}.

To the best of our knowledge, heuristic algorithms for the BQP were never investigated thoroughly in the the literature except some results on variations of block coordinate descent type algorithm applied to BQP$(Q, \zerovector{m}, \zerovector{n})$, where \zerovector{k} is a zero vector in $\mathbb{R}^k$. In this paper, we examine BQP systematically from an algorithmic point of view. In particular, we present very fast construction heuristics and more involved improvement heuristics. We show that a greedy type algorithm guarantees a $\frac{1}{m-1}$- optimal solution for BQP in polynomial time for $m > 2$. For $m\leq 2$ the algorithm is shown to produce an optimal solution. We also introduce various new neighborhoods for BQP which can be integrated into sophisticated search algorithms. The power of these neighborhoods is examined experimentally within a multi-start local search framework.   It is shown that a partitioning problem associated with one of our neighborhoods is NP-hard. However, approximations to this partitioning problem can be used effectively to define related neighborhoods that work well in practice. Our experimental analysis provides additional insights into the problem structure and properties various types of problem instances. Further, our work provide a systematically developed test-bed that can be used as benchmark for future research. Results of extensive experimental analysis are also provided using our algorithms and these are compared with the best known heuristic for solving QP~\citep{Wang2012}. This comparison confirms the need for developing special purpose algorithms for BQP.

The paper is organized as follows.  In Section~\ref{sec:construction}, we present several algorithms for quick construction of BQP solutions. We also provide theoretical analysis on  the performance of the greedy algorithm.  In Section~\ref{sec:improvement}, we propose more advanced heuristic approaches to solve the problem and establish theoretical properties of some of our neighborhoods. We also provide efficient implementation details, data structures, and complexity analysis of all our algorithms. Details of the testbed used in our experimental analysis is described in Section~\ref{sec:testbed}.  Results of extensive computational experiments using our algorithms and comparison with one of the best-known heuristic for QP~\citep{Wang2012} are presented in Section~\ref{sec:experiments}.  Finally, the concluding remarks are provided in Section~\ref{sec:conclusion}.

\section{Construction Heuristics}
\label{sec:construction}

Let us start from a general observation.  Assume that for some $i \in I$ both $c_i \ge 0$ and $q_{ij} \ge 0$ for all $j \in J$.  Then the value of $x_i$ can be fixed to 1, and this will preserve the optimal solutions.  If $c_i \le 0$ and $q_{ij} \le 0$ for all $j \in J$, then $x_i$ can be fixed to 0.  Similar results can be obtained for the $y_j$ variables.  Thus, all the NP-hard instances have mixed positive and negative values of $q_{ij}$ and/or $c_i$ and $d_j$.

Hence, the objective value of a feasible solution may be either negative or positive.  However, the optimal objective value is always non-negative since a \emph{trivial solution} $(x, y) = (\zerovector{m}, \zerovector{n})$ achieves $f(x, y) = 0$ for any problem instance.  A trivial solution can be used as a starting point for an improvement heuristic.  However, it is worth noting that this solution may turn out to be a deep local maximum for some local search neighborhoods and, thus, should be used carefully.

In order to obtain several different starting points, one can use \emph{random} solutions.  A random solution $(x, y)$ is obtained by choosing $x_i$ and $y_j$ randomly for each $i$ and $j$, respectively.  Observe that the expected value $\mathbb{E}[f]$ of a random solution is
\begin{equation}
\label{eq:expected-random-solution}
\mathbb{E}[f] = n_1 m_1 \overline{q} + m_1 \overline{c} + n_1 \overline{d},
\end{equation}
where $n_1$ and $m_1$ are the expected numbers of 1's in $y$ and $x$, respectively, and $\overline{q}$, $\overline{c}$ and $\overline{d}$ are the averages of $q$, $c$ and $d$.  The values $n_1$ and $m_1$ can be calculated as $n_1 = \mathbb{E}[\sum_{j \in J} y_j] = n \cdot p(y_j = 1)$ and $m_1 = \mathbb{E}[\sum_{i \in I} x_i] = m \cdot p(x_i = 1)$.  Thus, if $\overline{q} < 0$, $\overline{c} \le 0$ and $\overline{d} \le 0$, the expected objective value $\mathbb{E}[f]$ of a random solution will be negative.  An attempt to improve such a solution with a simple local search will usually generate a trivial solution.

Thus, a better approach is to use the following construction heuristic.  Let $\mathit{positive}(v) = v$ if $v \ge 0$ and $\mathit{positive}(v) = 0$ otherwise.  Let $w^+_i = c_i + \sum_{j \in J} \mathit{positive}(q_{ij})$.
Order the rows of the problem such that $w^+_i \ge w^+_{i+1}$ for $i = 1, 2, \ldots, m - 1$.  On the $i$th iteration, choose the best value of $x_i$ with the assumption that $x_1, x_2, \ldots, x_{i - 1}$ are fixed, $x_{i + 1} = x_{i + 2} = \ldots = x_m = 0$ and $y$ is selected optimally.  Note that the latter can be done efficiently since an optimal value of $y = y(x)$ given a fixed $x$ is as follows~\citep{Punnen2012}:
\begin{equation}
\label{eq:optimal-y}
y(x)_j = \begin{cases}
1 & \text{if } j \in J \connect{and} \displaystyle{\sum_{i \in I} q_{ij} x_i + d_j > 0}, \\
0 & \text{otherwise,}
\end{cases}
\end{equation}
We call this algorithm \emph{Greedy} and the solution produced by the Greedy algorithm, a \emph{Greedy solution}.  Our implementation of the Greedy heuristic (see Algorithm~\ref{alg:greedy}) terminates in $O(mn)$ time.

\begin{algorithm}[ht]
	Order the rows of the problem such that $w^+_i \ge w^+_{i+1}$ for $i \in I \setminus \{ m \}$\;
	$s_j \gets d_j$ for $j \in J$\;
	\For {$i \gets 1$ \KwTo $m$}
	{
		$f_0 \gets \sum_{j \in J} \mathit{positive}(s_j)$\;
		$f_1 \gets c_i + \sum_{j \in J} \mathit{positive}(s_j + q_{ij})$\;
		\lIf {$f_0 \ge f_1$}
		{
			$x_i \gets 0$\;
		}
		\Else
		{
			$x_i \gets 1$\;
			$s_j \gets s_j + q_{ij}$ \KwSty{for each} $j \in J$\;
		}
	}

	\BlankLine
	\For {$j \gets 1$ \KwTo $n$}
	{
		\lIf {$s_j > 0$}
		{
			$y_i \gets 1$\;
		}
		\lElse {$y_i \gets 0$}\;
	}
\caption{The Greedy algorithm implementation.}
\label{alg:greedy}
\end{algorithm}

Below, we provide some properties of the Greedy algorithm.

\begin{remark}
For $m \le 2$ and arbitrary $n$, the Greedy algorithm produces an optimal solution.
\end{remark}
\begin{proof}
Let $m = 1$.  Then the Greedy algorithm tests all possible values of $x$ and for each of those values it finds the optimal $y$, yielding an optimal solution.

Let $m = 2$.  Observe that the Greedy algorithm selects the best of solutions $(x, y)$, where $y \in \{ 0, 1 \}^n$ and $x$ is $x \in \big\{ (0, 0)^T \comma (1, 0)^T \comma (1, 1)^T \big\}$.  Observe also that if there exists an optimal solution $(x, y)$ such that $x = (0, 1)^T$, then $f(x, y) = w^+_2$ and, hence, there exists another optimal solution $(x', y')$ such that $x = (1, 0)^T$ and $f(x', y') = w^+_1 = w^+_2$. This establishes the theorem.
\end{proof}

\begin{theorem}
For $m > 2$, the Greedy algorithm provides a $\frac{1}{m - 1}$-approximation of the optimal solution, and this bound is sharp.
\end{theorem}
\begin{proof}
Let $w^+_i = c_i + \sum_{j \in J} \mathit{positive}(q_{ij})$ and the rows to be ordered such that $w^+_i \ge w^+_{i+1}$ for $i = 1, 2, \ldots, m - 1$.  We assume that $w^+_1 > 0$ as otherwise the Greedy algorithm produces a trivial solution which is optimal.

Let $(x^*, y^*)$ be an optimal solution of the problem.  Assume $x^*_i = 0$ for some $i \in I$.  Then the optimal objective $f(x^*, y^*)$ is at most $w^+_1 \cdot (m - 1)$.  Thus, if $(x, y)$ is a Greedy solution,
\[
\frac{f(x, y)}{f(x^*, y^*)} \ge \frac{w^+_1}{w^+_1 \cdot (m - 1)} = \frac{1}{m - 1} \,.
\]
Now assume that $x^*_i = 1$ for each $i \in I$.  Let $w^*_i = \sum_{j \in J} q_{ij} y^*_j$.  Recall that, after two iterations, the Greedy algorithm produces an optimal solution for the problem consisting of the first two rows.  Thus, $w^*_1 + w^*_2 \le f(x, y)$.  At the same time, $f(x, y) \ge w^+_1 \ge w^+_i \ge w^*_i$ for any $i \in I$ and, hence,
\[
\frac{f(x, y)}{f(x^*, y^*)} \ge \frac{f(x, y)}{f(x, y) + (m - 2)f(x, y)} = \frac{1}{m - 1} \,.
\]

To establish that the $\frac{1}{m - 1}$ bound is sharp, consider the following example, where $m = n$ and $c$ and $d$ are zero vectors:
\[
Q = \left[ \begin{array}{ccccc}
1 & -n & -n & \cdots & -n \\
-n & 1 & 0 & \cdots & 0 \\
-n & 0 & 1 && 0 \\
\vdots & \vdots && \ddots \\
-n & 0 & 0 && 1
\end{array} \right] \,.
\]
Observe that $(x, y)$ provides an optimal solution with objective value $f(x, y) = m - 1$ if $x = y = (0, 1, 1, \ldots, 1)^T$.  Indeed, if $x_1 = 1$, then $y_2 = y_3 = \ldots = y_n = 0$ and $f(x, y) \le 1$.  Since the Greedy algorithm fixes $x_1 = 1$ on its first iteration, it yields a solution with objective value 1.  The result follows.
\end{proof}

\section{Improvement Algorithms}
\label{sec:improvement}

The Greedy heuristic proposed above is usually very fast and produce solutions of acceptable quality.  One can further improve the quality of the solutions obtained using some of the algorithms discussed in this section.

\subsection{Alternating Algorithm}

This approach is well-known as block coordinate descent algorithm in the non-linear optimization literature and was used used to solve the BQP$(Q, \zerovector{m}, \zerovector{n})$ problem.  We provide a brief description of the algorithm in the context of BQP$(Q, c, d)$ and discuss ways to improve its performance.

Observe that, analogous to~(\ref{eq:optimal-y}), given a fixed $y$, one can efficiently compute the optimal $x = x(y)$~\citep{Punnen2012} using
\begin{equation}
\label{eq:optimal-x}
x(y)_i = \begin{cases}
1 & \text{if } i \in I \connect{and} \displaystyle{\sum_{j \in J} q_{ij} y_i + c_i > 0}, \\
0 & \mbox{otherwise.}
\end{cases}
\end{equation}
The \emph{Alternating Algorithm} is a simple heuristic that alternatively fixes $x$ and $y$ to improve $y$ and $x$, respectively, until a local optimum is reached. A formal description of this algorithm is given in Algorithm~\ref{alg:alternate-basic}.
\begin{algorithm}[ht]
\SetKwBlock{Loop}{loop}{endloop}
	$\lambda \gets 1$\;
	\Loop
	{
		\If {$y(x) \neq y$}
		{	
			$y \gets y(x)$\;
		}
		\lElseIf {$\lambda = 0$} {\KwSty{break}\;}
		$\lambda \gets 0$\;
		
		\If {$y(x) \neq y$}
		{
			$y \gets y(x)$\;
		}
		\lElse {\KwSty{break}\;}
	}
\caption{A naive implementation of the Alternating local search.}
\label{alg:alternate-basic}
\end{algorithm}

It takes $O(nm)$ time for the Alternating algorithm to explore the neighborhood $N_\text{alt}(x^0, y^0) = \big\{ (x, y^0) \suchthat x \in \{ 0, 1 \}^m \big\} \cup \big\{ (x^0, y) \suchthat y \in \{ 0, 1 \}^n \big\}$.  It is easy to verify that $|N_\text{alt}(x^0, y^0)| = 2^m + 2^n - 1$.

\begin{theorem}
\label{th:alternate}
If a solution $(x, y)$ is a local maximum in $N_\text{alt}(x, y)$, then $f(x, y) \ge 0$.
\end{theorem}
\begin{proof}
Consider a trivial solution $(\textbf{0}^m, \textbf{0}^n)$.  Observe that $(\zerovector{m}, \zerovector{n}) \in N_\text{alt}(x, y)$ for any solution $(x, y)$.  Since $f(\zerovector{m}, \zerovector{n}) = 0$ for any BQP instance, the objective value of the best solution in $N_\text{alt}(x, y)$ is at least 0.
\end{proof}

Although the general framework of the alternating algorithm is well known, we provide here additional improvements yielding better average running time.  We observed  that, in practice, the number of iterations of the Alternating heuristic may be noticeable but, after a couple of iterations, only a few variables $x_i$ and $y_j$ are getting updated.  By maintaining two arrays, namely $s$ and $w$, one can reduce the time needed to calculate $x(y)$ and $y(x)$ to $O(m)$ are $O(n)$ time, respectively, see Algorithm~\ref{alg:alternate}.  Updating the arrays $s$ and $w$, though, takes $O(mn)$ time in the worst case but it is significantly faster if only a small number of values of $x_i$ or $y_j$ are modified.  Thus, the improved implementation provides significantly better average performance while preserving the worst case time complexity.  Observe also that such an implementation is much more friendly than Algorithm~\ref{alg:alternate-basic} with regards to CPU cache~\citep{Karapetyan2009}.


\begin{algorithm}[ht]
	$s_j \gets d_j + \sum_{i \in I} q_{ij} x_i$ \KwSty{for each} $j \in J$\;
	$w_i \gets c_i + \sum_{j \in J} q_{ij} y_j$ \KwSty{for each} $i \in I$\;
	$\lambda \gets -1$\;
	\While {$\lambda \le 0$}
	{
		$\lambda \gets \lambda + 1$\;
		\For {$j \gets 1$ \KwTo $n$}
		{
			\uIf {$y_j = 0$ \KwSty{and} $s_j > 0$}
			{
				$y_j \gets 1$; $\lambda \gets 0$\;
				$w_i \gets w_i + q_{ij}$ \KwSty{for each} $i \in I$\;
			}
			\ElseIf {$y_j = 1$ \KwSty{and} $s_j < 0$}
			{
				$y_j \gets 0$; $\lambda \gets 0$\;
				$w_i \gets w_i - q_{ij}$ \KwSty{for each} $i \in I$\;				
			}
		}
		\BlankLine
		\lIf {$\lambda = 1$} {\KwSty{break}\;}
		$\lambda \gets 1$\;
		\For {$i \gets 1$ \KwTo $m$}
		{
			\uIf {$x_i = 0$ \KwSty{and} $w_i > 0$}
			{
				$x_i \gets 1$; $\lambda \gets 0$\;
				$s_j \gets s_j + q_{ij}$ \KwSty{for each} $j \in J$\;
			}
			\ElseIf {$x_i = 1$ \KwSty{and} $w_i < 0$}
			{
				$x_i \gets 0$; $\lambda \gets 0$\;
				$s_j \gets s_j - q_{ij}$ \KwSty{for each} $j \in J$\;
			}
		}
	}
\caption{An efficient implementation of the Alternating local search.}
\label{alg:alternate}
\end{algorithm}

\subsection{Portions-Based Algorithms}

Another approach to improve a BQP solution is to fix several variables $x_i$ and  solve the resulting constrained problem.  Let BQP$(Q, c, d, I^*, x^0)$ be defined as follows:
\begin{align*}
\text{maximize } & x^TQy + cx + dy\\
\text{subject to } & x_i = x^0_i \quad \text{for } i \notin I^*,\\
 & x \in \{ 0, 1 \}^m, y \in \{ 0, 1 \}^n.
\end{align*}
Observe that BQP$(Q, c, d, I^*, x^0)$ can be reduced to \linebreak[4] BQP$(Q^*, c^*, d^*)$, where $Q^* = (q_{ij})$ is a $k \times n$ matrix, $c$ is a vector in $\mathbb{R}^k$, $d$ is a vector in $\mathbb{R}^n$ and $k = |I^*|$.  Indeed, let $I^* = \{ \pi(1), \pi(2), \ldots, \pi(k) \}$ and, for $i = 1, 2, \ldots, k$, let
\[
q^*_{ij} = q_{\pi(i), j}, \quad c^*_i = c_{\pi(i)} \quad \text{and} \quad d^*_j = d_j + \sum_{i \in I \setminus I^*} q_{ij} x^0_i \,.
\]
Let $(x^*, y^*)$ be an optimal solution of BQP$(Q^*, c^*, d^*)$, where $Q^* = (q^*_{ij})$ is an $|I^*| \times n$ matrix.  Define
\[
x_i = \begin{cases}
x^*_r & \text{if } i = \pi(r) \connect{for some} r,\\
x^0_i & \text{otherwise}
\end{cases} \quad \text{and} \quad y = y^* \,.
\]
Then $(x, y)$ is an optimal solution to BQP$(Q, c, d, I^*, x^0)$.

The corresponding neighborhood $N_{I^*}(x^0, y^0) = \big\{ (x, y) \suchthat x \in \{ 0, 1 \}^m \comma y \in \{ 0, 1 \}^n \comma x_i = x^0_i \connect{for} i \notin I^* \big\}$ is of size $|N_{I^*}(x^0, y^0)| = 2^{|I^*| + n}$ and can be explored in $O(mn + n 2^{|I^*|})$ time as it takes $O(mn)$ time to produce a reduced problem BQP$(q^*, c^*, d^*)$ and $O(n2^{I^*})$ time to solve it.  Observe that, if $|I^*|$ is $O(\log m)$, the complexity of the first stage of the algorithm may dominate that of the second stage.  When several neighborhoods $N_{I^*}(x, y)$ are to be explored, the total time complexity of the first stage can be reduced.  In particular, by precalculating $s_j = d_j + \sum_{i \in I} q_{ij} x_i$ and further maintaining these values, one can explore $\kappa$ neighborhoods $N_{I^*}(x, y)$ in $O(mn + \kappa n 2^{|I^*|})$ time.  Hereafter, we assume this technique to be used in all implementations of the heuristics derived from this idea.

We now propose several improvement algorithms exploiting this approach.  Consider \emph{Exhaustive Portions} neighborhood defined as follows:
\[
N^k_\text{ex}(x, y) = \bigcup_{I^* \subset I, |I^*| = k} N_{I^*}(x, y) \,.
\]
Observe that the intersection $N_{I^*}(x, y) \cap N_{I^{**}}(x, y)$, where $I^*, I^{**} \subset I$, may be significant.  To avoid re-exploring candidate solutions, we start from enumerating all the subsets $I^* \subset I$ of size 1, then we enumerate all the subsets $I^* \subset I$ of size 2, etc.  For each subset, we only need to test one value of $x = x^*$, where
\begin{equation}
\label{eq:xportions}
x^*_i = \begin{cases}
1 - x_i & \text{if } i \in I^*,\\
x_i & \text{otherwise.}
\end{cases}
\end{equation}
For details see Algorithm~\ref{alg:exhaustive}.
\begin{algorithm}[ht]
	$\lambda \gets 1$\;
	\While {$\lambda = 1$}
	{
		$\lambda \gets 0$\;
		\For {$p \gets 1$ \KwTo $k$}
		{
			$\mu \gets 0$\;
			\While {$\mu < {m \choose p}$}
			{
				Select next $I^* \subset I$, $|I^*| = p$\;
				Calculate $x^*$ according to (\ref{eq:xportions})\;
				\If {$f(x^*, y(x^*)) > f(x, y(x))$}
				{
					$x \gets x^*$\;
					$\mu \gets 0$\;
					\lIf {$p > 1$} {$\lambda \gets 1$\;}
				}
				\lElse
				{
					$\mu \gets \mu + 1$\;
				}
			}
			
			\lIf {$\lambda = 1$} {\KwSty{break}\;}
		}
	}
\caption{An efficient implementation of the Exhaustive Portions local search.}
\label{alg:exhaustive}
\end{algorithm}

With the above algorithm, the Exhaustive Portions neighborhood is of size $|N^k_\text{ex}(x, y)| = 2^n \cdot \sum_{p = 0}^k {m \choose p}$ and can be explored in $O({m \choose k} n) \subset O(m^k n)$ time.  Thus, for a reasonably large instance, it is normally possible to use the Exhaustive Portions local search only with a very small $k$.  When $k = 1$, the algorithm becomes a simple local search that we call \emph{Flip} heuristic.  The corresponding neighborhood $N_\text{flip}(x^0, y^0) = \big\{ (x^0, y^0) \big\} \cup \big\{ (x, y) \suchthat x \in \{ 0, 1 \}^m \comma y \in \{ 0, 1 \}^n \connect{and} \sum_{i \in I} |x_i - x^0_i| = 1 \big\}$ is of size $|N_\text{flip}(x, y)| = m 2^n + 1$ and can be explored in $O(mn)$ time.  Observe that $N_\text{flip}(x, y)$ is larger than $N_\text{alt}(x, y)$ although the time complexity of the corresponding exploration algorithms is the same.  This may indicate superior quality of the Flip heuristic.

It may be noted that $N_\text{flip}(x^0, y^0)$ is not exactly equal to $N^1_\text{ex}(x^0, y^0)$.  Observe that most of the BQP heuristics yield solutions such that $y$ is already optimal for the given $x$.  Thus, in our implementation of the Flip heuristic, we avoid exploration of solutions that are different from the given one only in $y$.

Another strategy is to consider only a limited number of subsets $I^* \in I$, each of a larger size $k$.  The resulting algorithm is called \emph{Random Portions} improvement heuristic which proceeds as follows.  Let $k$, $1 < k \le m$, be a parameter of the algorithm.  Randomly select a set $I^* \subset I$ such that $|I^*| = k$.  Replace the solution $(x, y)$ with the best solution in $N_{I^*}(x, y)$ and proceed to the next random set $I^*$.  Terminate when either the prescribed number of iterations is completed or the given time limit is reached.  The time complexity of the Random Portions algorithm is $O(mn + \kappa n 2^k)$, where $\kappa$ is the number of iterations, and the size of the explored neighborhood $N^{k,\kappa}_\text{port}(x^0, y^0)$ is bounded by $|N^{k,\kappa}_\text{port}(x^0, y^0)| \le \kappa 2^{k + n}$.

\subsection{Variable Neighborhood Descent}

Given sufficiently large $k$, the Random Portions heuristic may be effective in removing a solution from a deep local maximum but it may be unnecessarily slow when applied to low quality solutions.  Thus, a good strategy is to start from a fast heuristic and, when it reaches a local maximum, apply the Random Portions algorithm for an appropriate value of $k$.  If it succeeds in removing the solution from the local maximum, go back to the fast heuristic.  Otherwise proceed to the Random Portions algorithm with a larger value of $k$.  Such an algorithm can be viewed as an implementation of the Variable Neighborhood Descent (VND)~\citep{Hansen2003}, and we implemented it as discussed in Algorithm~\ref{alg:vnd}.

\begin{algorithm}[ht]
	Produce a solution $(x, y)$ with the Greedy algorithm\;
	$\lambda \gets 1$\;
	\While {$\lambda = 1$}
	{
		$\lambda \gets 0$\;
		Improve the solution $(x, y)$ with the Alternating local search\;
		Improve the solution $(x, y)$ with the Flip local search\;
		\If {the solution was improved by the Flip local search}{$\lambda \gets 1$\;}
		\Else
		{
			\For {$k \gets 2$ \KwTo $p$}
			{
				\For {$\ell \gets 1$ \KwTo $m$}
				{
					Select a random $I^* \subseteq I$ such that $|I^*| = k$ and $\ell \in I^*$\;
					Fix $x_i$, $i \notin I^*$ and solve the remaining problem exactly\;
					If the solution was improved, set $\lambda \gets 1$\;
				}
			}
		}
	}
	
	\Return {the objective of $(x, y)$ and $(x, y)$}\;
\caption{The VND heuristic.}
\label{alg:vnd}
\end{algorithm}

In our VND, the initial solution is constructed with the Greedy heuristic.  Then it is repeatedly improved with Alternating and Flip heuristics until a local maximum in $N_\text{alt}(x, y) \cup N_\text{flip}(x, y)$ is reached.  Next, the Random Portions ($k = 2$) local search is applied.  If it does not succeed, the Random Portions ($k = 3$) local search is applied, etc.  If at some point the solution is improved, the algorithm restarts with Alternating and Flip heuristics.

The VND neighborhood $N_\text{VND}(x, y)$ is defined as follows:
\begin{equation}
\label{eq:vnd}
N_\text{VND}(x, y) = N_\text{alt}(x, y) \cup N_\text{flip}(x, y) \cup \bigcup_{p = 2}^k N^{p,m}_\text{port}(x, y) \,.
\end{equation}
Let us calculate the intersections of the components of (\ref{eq:vnd}).  Observe that the intersection of the first two neighborhoods is small: $N_\text{alt}(x^0, y^0) \cap N_\text{flip}(x^0, y^0) = \big\{ (x^0, y^0) \big\} \cup \big\{ (x, y^0) \suchthat x \in \{ 0, 1 \}^m \sum_{i \in I} |x_i - x^0_i| = 1 \big\}$.  In contrast, $N_\text{alt}(x^0, y^0) \cap N_{I^*}(x^0, y^0) = \big\{ (x^0, y) \suchthat y \in \{ 0, 1 \}^n \big\} \cup \big\{ (x, y^0) \suchthat x \in \{ 0, 1 \}^m \comma x_i = x^0_i \connect{for} i \notin I^* \big\}$, i.e., $N^{k,m}_\text{port}(x, y)$ includes a large part of $N_\text{alt}(x, y)$.  Moreover, $N_\text{flip}(x, y) \subset N^{k,m}_\text{port}(x, y)$ if the Random Portions heuristic is implemented as in Algorithm~\ref{alg:vnd}.

Let us now estimate the size of $N^{p,m}_\text{port}(x, y) \cap N^{k,m}_\text{port}(x, y)$.  In particular, we will show that the probability of choosing a set $I^*$ of size $p < k$ such that $N_{I^*}(x, y) \subset N^{k,m}_\text{port}(x, y)$ is low.  Observe that
$$
P(A \subset B) = \frac{k}{m} \cdot \frac{k - 1}{m - 1} \cdots \frac{k - p + 1}{m - p + 1} = \frac{k! \cdot (m - p)!}{m! \cdot (k - p)!} \,,
$$
if $A, B \in I$ are selected randomly, $|A| = p$ and $|B| = k$.  Thus, the probability that $N^{I^*}(x, y) \subset N^{k,m}_\text{port}(x, y)$ if $|I^*| = p$ is
\begin{equation}
\label{eq:probability}
P\big(N^{I^*}(x, y) \subset N^{k,m}_\text{port}(x, y)\big) = 1 - \left( 1 - \frac{k! \cdot (m - p)!}{m! \cdot (k - p)!} \right)^m \,.
\end{equation}
Note that, in Algorithm~\ref{alg:vnd}, the sets $I^*$ are not exactly arbitrary but for simplicity we can ignore it.  For small $k$ and $p$, $2 \le p < k$, we have
\begin{align*}
\lim_{m \to \infty} & P\big(N^{I^*}(x, y) \subset N^{k,m}_\text{port}(x, y)\big) \\
&= \lim_{m \to \infty} 1 - \left( 1 - \frac{k!}{(k - p)!} \cdot \frac{1}{m^p} \right)^m \\
&= \lim_{m \to \infty} 1 - \exp\left(- \frac{k!}{(k - p)! m^{p-1}}\right) = 0 \,.
\end{align*}
Even if $m = 100$, $k = 3$ and $p = 2$, the probability (\ref{eq:probability}) is below 6\%, i.e., only a small number of neighborhoods $N_{I^*}(x, y)$ dominated by $N^{k,m}_\text{port}(x, y)$ are selected within \linebreak[4]$N^{p,m}_\text{port}(x, y)$.

Thus, the neighborhood $N_\text{VND}(x, y)$ is significantly larger than any of its components which indicates the potential of the VND heuristic to yield superior solutions.  Although the quick heuristics Alternating and Flip are mostly dominated by following algorithms, it is expected that they speed up the VND method.

\bigskip

Another variation of VND that we call \emph{VND Exhaustive} alternates between the Alternating heuristic and the Exhaustive Portions algorithm.  As for the Flip heuristic, our implementation of VND Exhaustive does not explore solutions $\big\{ (x^0, y) \suchthat y \in \{ 0, 1 \}^n \big\}$ within the Exhaustive Portions local search.


\subsection{Multi-Start Algorithms}

Multi-Start method is a simple yet efficient technique that have been applied to many combinatorial optimization problems~\citep{Marti2003}.  There exist several variations of the metaheuristic; the one that we used is as follows:
\begin{inparaenum}[(a)]
\item produce a random solution;
\item improve it;
\item save the obtained solution if it is better than the best one found so far;
\item terminate if the given time has elapsed or proceed to the next random solution otherwise.
\end{inparaenum}

We have three fast local search algorithms, namely Alternating, Flip and VND Exhaustive ($k = 1$).  These algorithms are embedded within the multi-start framework to obtain enhanced heuristics.

\subsection{Row-Merge Algorithms}
\label{sec:rowmerge}

Another approach to reduce heuristically the size of the problem is to partition its rows into clusters and consider all the rows in each cluster merged together.  Let $\partition = \{ I^1, I^2, \ldots, I^k \}$ be a partition of $I$ into clusters $I^1$, $I^2$, \ldots, $I^k$.  Introduce a problem BQP$(Q, c, d, \partition)$:
\begin{align}
\text{maximize } & x^TQy + cx + dy\\
\text{subject to } & x_i = x_\ell \quad \text{for } i, \ell \in C,\ C \in \partition,\label{eq:bqp-partition-constraint}\\
 & x \in \{ 0, 1 \}^m, y \in \{ 0, 1 \}^n.
\end{align}
To solve BQP$(Q, c, d, \partition)$, compute
\[
q^*_{ij} = \sum_{r \in I^i} q_{rj}, \quad c^*_i = \sum_{r \in I^i} c_r \quad \text{and} \quad d^* = d
\]
for $i \in \{ 1, 2, \ldots, k \}$.  Let $(x^*, y^*)$ be an optimal solution of BQP$(Q^*, c^*, d^*)$, where $Q^* = (q^*_{ij})$ is a $k \times n$ matrix.  Let $x_r = x^*_i$ for each $r \in I^i$ and $i \in \{ 1, 2, \ldots, k \}$.  Then $(x, y^*)$ is an optimal solution to BQP$(Q, c, d, \partition)$.

We call this method \emph{Clustering Row-Merge} heuristic.  Its quality depends significantly on the partition $\partition$.  In particular, it is expected that rows $i$ and $\ell$ are `merged' only if $x_i = x_\ell$ in most of near-optimal solutions of the problem.  We propose the following partitioning technique.

Let $S^1$, $S^2$, \ldots, $S^p$ be $p$ `good solution' to BQP$(Q, c, d)$.  Create a complete weighted graph $G$ with a node set $I$\@.  Let the weight $w_{i \ell}$ of an edge $(i, \ell)$ in $G$ be the number of solutions $(x, y) \in \{S^1, S^2, \ldots, S^p\}$ such that $x_i = x_\ell$.  For a set of nodes $C \subseteq I$, let
\[
\mu(C) = \begin{cases}
\displaystyle{\min_{i \neq \ell \in C} w_{i \ell}} & \text{if $|C| > 1$,} \\
p & \text{otherwise}
\end{cases}
\]
and
\begin{equation}
\label{eq:partition-weight}
w(C) = |C| \cdot \mu(C) \,.
\end{equation}
Define the \emph{Partitioning Problem} be the problem of finding a partition $\partition = \{ I^1, I^2, \ldots, I^k \}$ of a given size $k$ such that $w(\partition) = \sum_{i = 1}^k w(I^k)$ is maximized.  Since the weight $w(I^i)$ of a cluster $I^i$ is proportional to the weight of the lightest edge in the clique induced by $I^i$, the elements of such a partition will tend to have no light edges.  It is important to note that $w(C)$ is also proportional to the cardinality of $C$ and, hence, it is preferable to have the number of clusters with light edges as small as possible.

\begin{theorem}
The Partitioning Problem is strongly NP-hard even if $k = 2$.
\end{theorem}
\begin{proof}
We reduce the maximum clique problem to the partitioning problem.  Let $G = (I, E)$ be a given graph.  Let $G' = (I', E')$ be a complete weighted graph with node set $I' = I \cup \{ v \}$.  Let $w_{i \ell} = p$ if $(i, \ell) \in E$ and $w_{i \ell} = 0$ otherwise, where $w_{i \ell}$ is a weight assigned to the edge $(i, \ell) \in E$.  Observe that in such a graph, $w(C) = |C|p$ if $C \subset I'$ is a clique and $w(C) = 0$ otherwise (see~\ref{eq:partition-weight}).  

Let $k = 2$ and $\partition = \{ I^1, I^2 \}$ be an optimal partition of $I'$.  For simplicity, assume $v \in I^1$.  Then there exist the following options:
\begin{enumerate}
	\item If $|I| = 1$, the problem is trivial, and the maximum clique in $G$ is $I$.
	\item If $w(\partition) = p$, the graph $G$ is an independent set.  Indeed, if there exists an edge $(i, \ell) \in E$, then setting $I^2 = \{ i, \ell \}$ and $I^1 = I' \setminus I^2$ would result in $w(\partition) = 2p$.  If $|I| = 1$, the partition $I^1 = \{ v \}$ and $I^2 = I$ provides a solution with objective value $2p$.
	\item If $w(\partition) = |I'|p$, then $G$ is complete.  Indeed, if $|I^1| > 1$, then $I^1$ is not a clique and, thus, $w(I^1) = 0$.  Provided that $w(I^2) \le |I^2|p < |I'|p$, we get $I^1 = \{ v \}$ and $w(I^1) = p$.  Therefore, $w(I^2) = |I| p$ and, thus, $I^2 = I$ is a clique in $G$, i.e., $G$ is complete.
	\item If $p < w(\partition) < |I'|p$, then $I^2$ is a maximum clique in $G$.  Indeed, if $I^2$ is not a clique in $G$, then $w(I^2) = 0$ and $w(\partition) \le p$ since $w(I^1) \le p$.  Thus, $I^2$ is a clique and $|I^1| > 1$ resulting in $w(I^1) = 0$.  Observe that $w(\partition) = |I^2|p$ and, therefore, $I^2$ is maximized.  The result follows.
	\item Note that $w(\partition) \ge p$ for any graph $G$ because the weight of a partition $I^1 = \{ v \}$ and $I^2 = I$ is at least $p$.
\end{enumerate}
\end{proof}

Since the Partitioning Problem is NP-hard, we use a heuristic approach to solve it.  Consider the \emph{Greedy Partitioning Algorithm} that proceeds as follows.  Create a partition $\partition = \big\{ \{ 1 \}, \{ 2 \}, \ldots, \{ m \} \big\}$.  In each iteration, choose a pair $P, Q \in \partition$ such that $w(\partition \cup \{ P \cup Q \} \setminus \{ P, Q \} )$ is maximized and update $\partition \gets \partition \cup \{ P \cup Q \} \setminus \{ P, Q \}$.  Repeat this until $|\partition| = k$.

The algorithm performs $m - k$ iterations.  In each iteration, there are $m^2$ pairs of elements of $\partition$ to be compared.  It takes $O(m^2)$ time to calculate the weight of a merged cluster.  Hence, the time complexity of the Greedy Partitioning Algorithm is $O(m^5)$, which is unreasonable even for moderate instances.  The following is a modification of the algorithm that terminates in $O(m^2 \log m)$ time.

For $P, Q \subset I$, let $\delta(P, Q) = w(P \cup Q) - w(P) - w(Q)$ and $\mu(P, Q) = \mu(P \cup Q)$.  Set $\mu(P) = p$ for each $P$ in the initial partition $\partition = \big\{ \{ 1 \}, \{ 2 \}, \ldots, \{ m \} \big\}$.  Calculate the values $\mu(P, Q)$ for each $P \neq Q \in \partition$ and place all these pairs in an ordered set $L$.  Sort $L$ by $\delta(P, Q)$ in descending order.  Let $(P, Q)$ be the first element in $L$.  Merge $P$ and $Q$ by updating $\partition \gets \partition \cup (P \cup Q) \setminus \{ P, Q \}$.  Remove $(P, Q)$ from $L$.  For each $R \in \partition \setminus \{ P \cup Q \}$, also remove $(R, P)$ and $(R, Q)$ from $L$ and insert $(R, P \cup Q)$ such that the ordering of $L$ is preserved.  Note that $\mu(P \cup Q) = \mu(P, Q)$ and $\mu(R, P \cup Q)$ can be calculated in $O(1)$ time as
\[
\mu(R, P \cup Q) = \min \big\{ \mu(P, Q), \mu(R, P), \mu(R, Q) \big\} \,.
\]
Repeat the procedure until $|\partition| = k$.  The first iteration of the algorithm takes $O(m^2 \log m)$ time.  Then, on each iteration, the algorithm spends $O(m \log m)$ time to update $L$ and the values of $\mu$.  Since the number of iterations is $m - k$, the total complexity of the algorithm is $O(m^2 \log m)$.

By exploiting the fact that $p$ is usually a small fixed number, we are able to further speed up the Greedy Partitioning Algorithm.  Observe that $-pm \le \delta(P, Q) \le 0$.  For each $w = -pm, -pm + 1, \ldots, 0$, introduce a list $a_w$ that contains all the pairs $(P, Q)$ such that $\delta(P, Q) = w$, where $P$ and $Q$ are in the initial partition.  Creating such lists takes only $O(m^2)$ time.  On every iteration, select an arbitrary pair of cliques $(P, Q) \in a_w$, where $w$ is the largest index such that $a_{w} \neq \emptyset$.  Since the number of lists $a_w$ is $O(m)$, this operation takes only $O(m)$ time.  The updating procedure is similar to the one described above and also takes $O(m)$ time.  Thus, the complexity of the algorithm is $O(m^2)$.  We used this implementation in the Clustering Row-Merge algorithm.

Observe that both implementations of the Greedy Partitioning Algorithm require quick operations with pairs $(P, Q)$.  In our implementation, we maintain a hash table for each $P \in \partition$.  Such a hash table stores all the pairs $(P, Q)$, where $Q \in \partition$.  As the hash key, we use the first vertex $Q_1$ in $Q$.  Since the clusters are non-intersecting, such a key is unique and limited by $1 \le Q_1 \le m$.

\bigskip

The Clustering Row-Merge heuristic effectively exploits information obtained from several near-optimal solutions of the problem and, thus, may be useful in evolutionary algorithms and other metaheuristics.  However, if such solutions are unavailable, the Clustering Row-Merge algorithm cannot be used.  In order to apply the idea of row merging in a standalone heuristic, we propose the following method.  Generate a random partition $\partition = \{ I^1, I^2, \ldots, I^k \}$ of $I$ such that $I^i \neq \emptyset$ for each $i$.  Solve BQP$(Q, c, d, \partition)$ as described above and improve the resulting solution with a fast local search.  If the obtained solution is better than the best known so far, save it.  Repeat the procedure until the given time is elapsed.  We call this method \emph{Multi-Start Row-Merge}.

The Multi-Start Row-Merge algorithm can be modified to be used as an improvement procedure.  Let $(x, y)$ be the starting solution.  Generate a random partition $\partition(x)$ of $I$ of size $k$ such that $x_i = x_\ell$ for each $x, \ell \in C$, where $C \in \partition(x)$.  A solution to BQP$(Q, c, d, \partition(x))$ is the best solution in the neighborhood $N_\text{merge}(x^0, y^0) = \big\{ (x, y) \suchthat x \in \{ 0, 1 \}^m \comma x_i = x_\ell \connect{if} i, \ell \in C \connectcomma{where} C \in \partition(x^0) \connectcomma{and} y \in \{ 0, 1 \}^n \big\}$.  The size of this neighborhood is $|N_\text{merge}(x^0, y^0)| = 2^{k+n}$ and it takes $O(n2^k)$ time to explore it.  Observe that $(x, y) \in N_\text{merge}(x, y)$ and, thus, the corresponding local search never worsens the solution.  We call this improvement procedure \emph{Row-Merge Local Search}.

Assume $(x, y) \in N_\text{merge}(x^0, y^0)$ is an improvement over $(x^0, y^0)$.  Then $x$ is different from $x^0$ in at least $|C|$ values, where $C \in \partition(x^0)$.  Intuitively, it is unlikely that such an improvement is possible if $(x^0, y^0)$ is a near-optimal solution and the cluster $C$ is large.  Thus, this local search can be efficient only if $k$ is large.  Hence, it is impractical to solve BQP$(Q, c, d, \partition(x))$ exactly.  Note, however, that a heuristic solution to BQP$(Q, c, d, \partition(x))$ may be worse than $(x^0, y^0)$ and, thus, the Row-Merge Local Search algorithm needs to evaluate a solution before accepting it.




\section{Testbed}
\label{sec:testbed}

There are no standard test problems available in the literature for BQP\@.  Thus, we have created a testbed consisting of five instance types that correspond to some  of the real life applications of the problem.

In order to generate some of the instances, we need random bipartite graphs.  To generate a random bipartite graph $G = (V, U, E)$, we define seven parameters, namely $m = |V|$, $n = |U|$, $\underline{d}_1$, $\bar{d_1}$, $\underline{d}_2$, $\bar{d_2}$ and $\mu$ such that $0 \le \underline{d}_1 \le \bar{d_1} \le n$, $0 \le \underline{d}_2 \le \bar{d_2} \le m$, $m \underline{d}_1 \le n \bar{d}_2$ and $m \bar{d}_1 \ge n \underline{d}_2$.

The bipartite graph generator proceeds as follows.
\begin{enumerate}
	\item For each node $v \in V$, set $d_v$ as a uniformly distributed random integer in the range $\underline{d}_1 \le d_v \le \bar{d}_1$.
	\item For each node $u \in U$, set $d_u$ as a uniformly distributed random integer in the range $\underline{d}_2 \le d_u \le \bar{d}_2$.
	\item While $\sum_{v \in V} d_v \neq \sum_{u \in U} d_u$, alternatively select a node in $V$ or $U$ and regenerate its degree as described above.\footnote{In practice, if $m (\underline{d}_1 + \bar{d}_1) \approx n (\underline{d}_2 + \bar{d}_2)$, this algorithm converges very quickly.  However, in theory it may not terminate in finite time.  Thus, one can update it as follows.  If the equality $\sum_{v \in V} d_v = \sum_{u \in U} d_u$ is not achieved after a certain number of iterations, start a fixing procedure as follows.  Let $\delta = 1$ if $\sum_{v \in V} d_v < \sum_{u \in U} d_u$ and $\delta = -1$ otherwise.  While $\sum_{v \in V} d_v \neq \sum_{u \in U} d_u$ and $m \underline{d}_1 \le \delta + \sum_{v \in V} d_v \le m \bar{d}_1$, select randomly $v \in V$ such that $\underline{d}_1 \le d_v + \delta \le \bar{d}_1$ and update $d_v \gets d_v + \delta$.  While $\sum_{v \in V} d_v \neq \sum_{u \in U} d_u$, select randomly $u \in U$ such that $\underline{d}_2 \le d_u - \delta \le \bar{d}_2$ and update $d_u \gets d_u - \delta$.}
	\item Create a bipartite graph $G = (V, U, E)$, where $E = \emptyset$.
	\item Select a node $v \in V$ such that $d_v > \deg v$ (if no such node exists, go to the next step).  Let $U' = \{ u \in U \suchthat \deg u < d_u \connect{and} (v, u) \notin E \}$.  If $U' \neq \emptyset$, select a node $u \in U'$ randomly.  Otherwise randomly select a node $u \in U$ such that $(v, u) \notin E$ and $d_u > 0$; randomly select a node $v' \in V$ adjacent to $u$ and delete the edge $(v', u)$.  Add an edge $(v, u)$.  Repeat this step.
	\item For each edge $(v, u) \in E$, the weight $w_{vu}$ is a normally distributed integer ($\sigma = 100$ and $\mu$ is given).
\end{enumerate}


The following are the instance types used in our computational experiments.

\begin{enumerate}
\item The \emph{Random} instances are as follows: $q_{ij}$, $c_i$ and $d_j$ are normally distributed random integers with mean $\mu = 0$ and standard deviation $\sigma = 100$.

\item The \emph{Max Biclique} instances model the problem of finding a biclique of the maximum weight in a bipartite graph.  Let $G = (I, J, E)$ be a random bipartite graph with $\underline{d}_1 = n / 5$, $\bar{d_1} = n$, $\underline{d}_2 = m / 5$, $\bar{d_2} = m$ and $\mu = 100$.  Note that setting $\mu$ to 0 would make the weight of any large biclique likely to be around 0 and, thus, the landscape of the problem would be rather flat.  If $w_{ij}$ is the weight of an edge $(i, j) \in E$, set $q_{ij} = w_{ij}$ for every $i \in I$ and $j \in J$ if $(i, j) \in E$ and $q_{ij} = -M$ otherwise, where $M$ is large number.  Set $c = \zerovector{m}$ and $d = \zerovector{n}$.

\item The \emph{Max Induced Subgraph} instances model the problem of finding a subset of nodes in a bipartite graph that maximizes the total weight of the resulting induced subgraph.  The Max Induced Subgraph instances are similar to the Max Biclique instances except that $q_{ij} = 0$ if $(i, j) \notin E$ and $\mu = 0$.  The latter change is needed because having too many edges with positive weights would make the problem very simple (indeed, including all or almost all the nodes would yield a good solution).

\item The \emph{MaxCut} instances model the MaxCut problem as follows.  First, we generate a random bipartite graph as for the Max Induced Subgraph instances.  Then, we set $q_{ij} = -2 w_{ij}$ if $(i, j) \in E$ and $q_{ij} = 0$ if $(i, j) \notin E$.  Finally, we set $c_i = \frac{1}{2} \sum_{j \in J} q_{ij}$ and $d_j = \frac{1}{2} \sum_{i \in I} q_{ij}$.  For an explanation, see~\citep{Punnen2012}.

\item The \emph{Matrix Factorization} instances model the problem of producing a rank one approximation of a binary matrix.  The original matrix $H = (h_{ij})$ (see Section~\ref{sec:introduction}) is generated randomly with probability 0.5 of $h_{ij} = 1$.  The values of $q_{ij}$ are then calculated as $q_{ij} = 1 - 2 h_{ij}$, and $c = \zerovector{m}$ and $d = \zerovector{n}$.

\end{enumerate}

The proposed instances are hard to solve exactly. We have optimal solutions for many small instances but we do not know optimal objective values of any of the moderate and large instances used in our computational experiments.  Thus, in our experiments, we use the best known objective values to evaluate the quality of the obtained solutions.  All our test instances and best known solutions can be found at \url{http://www.cs.nott.ac.uk/~dxk/}. The best known objective function value is taken from our extensive experiments during the fine-tuning stages. The experimental results presented in various tables in this are obtained for a fixed parameter settings of the corresponding algorithms. We plan to keep the best-known solutions periodically updated based on our own experiments and results provided by other researchers.

\section{Empirical Evaluation}
\label{sec:experiments}

In this section, we provide empirical analysis of the algorithms discussed in Sections~\ref{sec:construction} and \ref{sec:improvement}.  Our experiments were conducted on an Intel i7-2600 CPU based PC\@.  All the algorithms are implemented in C\# 4.0, and no concurrency is used.

Several of the heuristic techniques described in this paper require an exact algorithm to solve a reduced problem. We are not attempting to provide efficient exact algorithms here. Since the reduced problems are of small size, some rudimentary algorithms are sufficient for our purpose. Obviously, more sophisticated exact algorithms will speed up our heuristics but the solution quality remains unaltered when the number of iterations is fixed.  We tested two basic approaches to  obtain an optimal solution for small size BQP\@.  Our  first approach is called an \emph{Exhaustive Enumeration}.  Since an optimal $y$ for a fixed $x$ can be obtained efficiently, we only need to try all possible $x$ values in an enumerative scheme. There are $2^m$ such values to be considered.  Thus, with a careful implementation, such an exhaustive enumeration algorithm terminates in $O(n 2^m)$ time.  Our second algorithm uses a mixed integer programming solver applied to the following problem:
$$
\begin{array}{rll}
\text{maximize} & \multicolumn{2}{l}{\displaystyle{\sum_{i \in I} \sum_{j \in J} q_{ij} z_{ij} + \sum_{i \in I} c_i x_i + \sum_{j \in J} d_j y_j}} \\
\text{subject to} \\
& z_{ij} \le x_i & \connect{for} i \in I \connect{and} j \in J,\\
& z_{ij} \le y_j & \connect{for} i \in I \connect{and} j \in J,\\
& z_{ij} \ge x_i + y_j - 1 & \connect{for} i \in I \connect{and} j \in J,\\
& x_i \in \{ 0, 1 \} & \connect{for} i \in I,\\
& 0 \le y_j \le 1 & \connect{for} j \in J,\\
& 0 \le z_{ij} \le 1 & \connect{for} i \in I \connect{and} j \in J.
\end{array}
$$
We will refer to this as \emph{MIP} approach.

As a mixed integer programming solver, we used CPLEX 12.4.  Setting parameters appropriately to choose suitable branch and bound strategy in CPLEX is crucial.  We found that the feasibility emphasis setting usually provided better performance in our experiments.  For the other CPLEX parameters, we used the default values except that we disabled multithreading.  To produce starting solutions for the MIP approach, we used the Greedy algorithm followed by the Alternating local search.  Giving a starting solution is not necessary for CPLEX but we found that providing it speeds up the optimization process.

\begin{table*}[ht] \centering
\begin{tabular}{@{} l @{} c @{} r r r r r @{} c @{} r @{}}
\toprule
$m \times n$ &\hspace*{2em}&Random&Max Bicl.&Max Ind.&MaxCut&Matr.~Fact.&\hspace*{2em}&Exhaustive~En.\\
\midrule
$20 \times 50$&&2.2&0.2&0.8&4.2&23.2&&0.5\\
$25 \times 50$&&64.9&0.4&3.1&55.3&419.9&&15.0\\
$30 \times 50$&&217.6&1.3&62.4&721.5&4\,820.1&&475.5\\
$35 \times 50$&&10\,956.4&0.8&45.8&3\,413.6&---&&14\,211.0\\
$40 \times 50$&&---&4.2&124.4&---&---&&$\approx 5.5 \cdot 10^5$\\
$45 \times 50$&&---&7.4&691.3&---&---&&$\approx 1.8 \cdot 10^7$\\
$50 \times 50$&&---&35.7&2\,658.6&---&---&&$\approx 5.6 \cdot 10^8$\\
\bottomrule
\end{tabular}
\caption{Evaluation of exact algorithms.  For each instance, the running time, in seconds, is reported for the MIP approach (if it does not exceed 5~hours).  The last column reports the running time of the Exhaustive Enumeration algorithm, in seconds, for the Random instances of the given size (for the instances of size $m \ge 40$, the reported times are obtained according to (\ref{eq:enum-time})).  The running time of the MIP approach includes the time spent on generating starting solutions.}
\label{tab:compare-exact}
\end{table*}

The results of our experiments with these rudimentary exact algorithms are reported in Table~\ref{tab:compare-exact}.  We report the running times of the Exhaustive Enumeration method for the Random instances only because its performance practically does not depend on the instance type. On our computational platform, the running time of the Exhaustive Enumeration algorithm can be roughly approximated as
\begin{equation}
\label{eq:enum-time}
t_\text{exhaustive}(n) = 10^{-8} \cdot n 2^{m} \text{ sec.}
\end{equation}
In contrast, the behavior of the MIP approach depends significantly on the size and the type of the problem, see Table~\ref{tab:compare-exact}.  Note, for example, that many matrix elements in the Max Biclique instances are `forbidden', which allows the MIP solver eliminate numerous variables each time some variable $x_i$ or $y_j$ is fixed to 1.  This explains the success of the MIP approach with the Max Biclique instances.

Observe that the MIP approach usually outperforms the Exhaustive Enumeration algorithm, especially for larger instances.  This result, however, holds only for small values of $n$.  Indeed, the performance of the MIP approach dramatically depends on the number of columns while the performance of the Exhaustive Enumeration is linear in $n$, see Figure~\ref{fig:exact-n}.  As a result, Exhaustive Enumeration outperforms MIP approach even for the Max Biclique instances, provided $n \gg m$.  Since the reduced problems arising within Random Portions, VND and Row-Merge heuristics preserve the number of columns of the original instance, we used the Exhaustive Enumeration algorithm in our implementations when an exact solution was needed.

\begin{figure*}[ht]
\begin{tikzpicture}
	\begin{semilogyaxis}[
		width=\textwidth,
		height=35ex,
		legend pos=south east,
		xlabel={The number of columns $n$},
		ylabel={Running time, sec},
		cycle list={
			{black, very thick, mark=*},
			{black, dashed, very thick},
		},
		legend style={cells={anchor=west}}
	]
	\addplot+ coordinates {
		(20, 0.058398)
		(40, 0.205836)
		(60, 0.259033)
		(80, 0.791367)
		(100, 0.576414)
		(120, 1.651994)
		(140, 1.060171)
		(160, 1.607649)
		(180, 2.714969)
		(200, 5.680359)
		(220, 3.744218)
		(240, 2.816628)
		(260, 6.413883)
		(280, 5.355944)
		(300, 4.815085)
		(320, 16.817551)
		(340, 6.483451)
		(360, 10.673724)
		(380, 5.818882)
		(400, 11.1396)
		(420, 20.689898)
		(440, 13.307324)
		(460, 16.161302)
		(480, 25.501894)
		(500, 18.378439)
	};
	\addlegendentry{MIP}
	\addplot+ coordinates {
		(20, 0.233151)
		(40, 0.416211)
		(60, 0.561778)
		(80, 0.799886)
		(100, 0.993978)
		(120, 0.996402)
		(140, 1.25131)
		(160, 1.371699)
		(180, 1.525548)
		(200, 1.854732)
		(220, 1.977129)
		(240, 2.1801)
		(260, 2.316807)
		(280, 2.151658)
		(300, 2.797814)
		(320, 3.349567)
		(340, 3.716738)
		(360, 3.724546)
		(380, 3.963442)
		(400, 3.857103)
		(420, 4.150096)
		(440, 4.257872)
		(460, 3.766495)
		(480, 4.560709)
		(500, 5.115744)
	};
	\addlegendentry{Exhaustive Enumeration}

	\end{semilogyaxis}
\end{tikzpicture}
\caption{The performance of the MIP approach significantly depends on the number of columns $n$ while the running time of the Exhaustive algorithm is linear in $n$.  The experiment was conducted for Max Biclique instances, $m = 20$, $n = 20, 40, \ldots, 500$.}
\label{fig:exact-n}
\end{figure*}
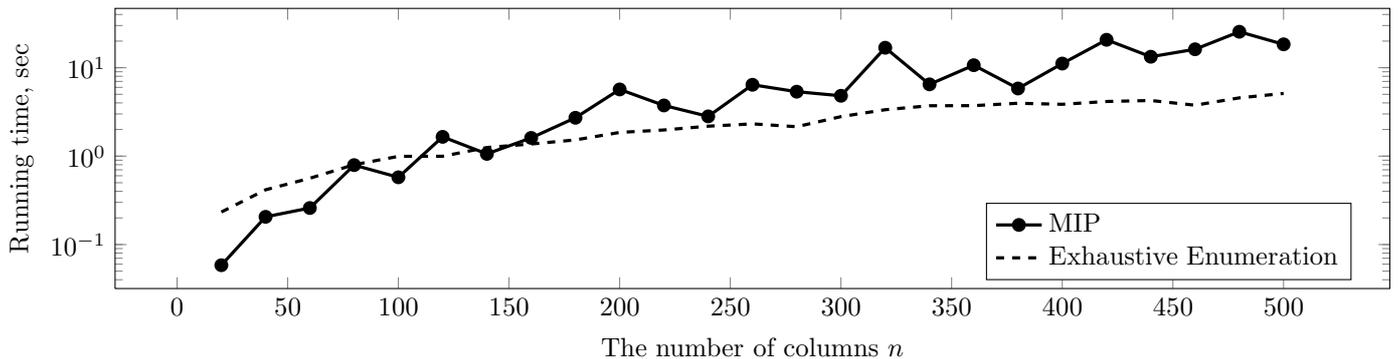

Another interesting observation is that the integer programming formulation cannot be used to obtain a good upper bound for the BQP\@.  The LP relaxation of the problem provides an objective that is normally several times larger than that of the optimal solution.  One can also use a mixed integer program solver to find an upper bound by giving it a limited time.  We noticed, however, that, in case of the BQP, the gap between the lower and the upper bounds is usually large until the very last steps of the algorithm, see Figure~\ref{fig:mip-dynamics}.
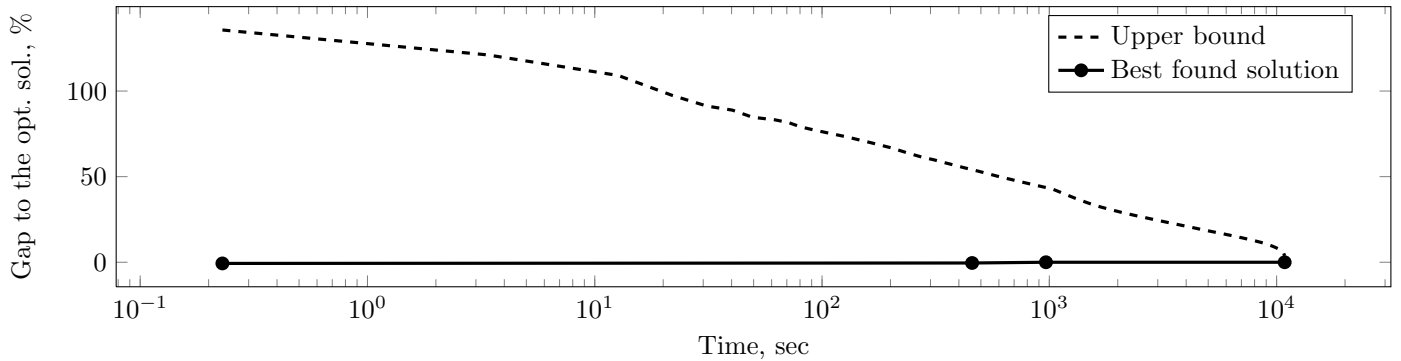
\begin{figure*}[ht]
\centering
\begin{tikzpicture}
\begin{semilogxaxis}[
xlabel={Time, sec},
ylabel={Gap to the opt.\ sol., \%},
width=\textwidth,
height=35ex,
legend pos=north east,
legend style={cells={anchor=west}}
]
\addlegendentry{Upper bound}
\addlegendentry{Best found solution}
\addplot[black,dashed,very thick] coordinates {
(0.23, 135.66779852858)
(3.37, 121.130447085456)
(12.68, 109.10441426146)
(21.95, 97.4320882852292)
(31.15, 91.2209960384833)
(40.34, 88.8440860215054)
(49.53, 84.6668081494058)
(61.9, 83.2059988681381)
(71.09, 81.6178551216752)
(80.31, 79.0994623655914)
(92.2, 77.2425014148274)
(128.76, 73.1819468024901)
(165.32, 69.708545557442)
(201.79, 66.8081494057725)
(238.23, 63.7945670628183)
(274.8, 61.4247311827957)
(311.22, 59.8082908885116)
(347.66, 58.0963497453311)
(383.89, 56.6461516694963)
(420.34, 55.2631578947369)
(457.39, 54.0923882286361)
(491.65, 53.0489530277306)
(529.01, 51.9418505942275)
(565.44, 50.9302490096208)
(601.85, 49.9858517260894)
(638.31, 49.1652518392756)
(674.66, 48.4543010752688)
(711.6, 47.6655348047538)
(748.45, 46.9934917940011)
(784.81, 46.3780418788908)
(821.17, 45.7590548953028)
(857.51, 45.1719015280136)
(893.82, 44.6661007357102)
(930.16, 44.1249292586304)
(966.49, 43.6332767402377)
(999.89, 43.1840690435767)
(1036.28, 42.6711941143181)
(1072.93, 41.8647425014148)
(1109.36, 41.0689020939445)
(1145.89, 40.2801358234295)
(1182.69, 39.5408885116016)
(1219.26, 38.8440860215054)
(1256.04, 38.168505942275)
(1292.72, 37.4929258630447)
(1329.68, 36.9128466327108)
(1365.34, 36.3327674023769)
(1401.79, 35.7597623089983)
(1438.5, 35.2185908319185)
(1475.18, 34.6915676287493)
(1512.07, 34.2034521788342)
(1548.48, 33.7542444821732)
(1585.95, 33.2838143746463)
(1623.95, 32.8805885681947)
(1662.95, 32.4844368986984)
(1702.92, 32.098896434635)
(1743.66, 31.7381154499151)
(1784.76, 31.3985568760611)
(1826.34, 31.0342388228636)
(1867.35, 30.7088285229202)
(1908.22, 30.3763440860215)
(1949.58, 30.0686191284663)
(1990.57, 29.7679683078664)
(2032.01, 29.4567062818336)
(2073.72, 29.1666666666667)
(2116.11, 28.8730899830221)
(2158.65, 28.6007357102434)
(2201.47, 28.3425297113752)
(2244.96, 28.1020090548953)
(2288.83, 27.8579513299377)
(2332.43, 27.6068194680249)
(2376.54, 27.366298811545)
(2421, 27.136389360498)
(2465.55, 26.9135540464063)
(2509.71, 26.7013299377476)
(2554.13, 26.4784946236559)
(2598.79, 26.2591963780419)
(2642.88, 26.0469722693831)
(2686.56, 25.8453593661573)
(2730.66, 25.6331352574986)
(2775.23, 25.4244482173175)
(2819.16, 25.2228353140917)
(2863.82, 25.0318336162988)
(2908.14, 24.8408319185059)
(2953.15, 24.6498302207131)
(2998.23, 24.4659026598755)
(3043.11, 24.2855121675156)
(3087.85, 24.1086587436333)
(3136.1, 23.9353423882286)
(3181.7, 23.7655631013016)
(3227.46, 23.592246745897)
(3271.22, 23.4260045274477)
(3317.25, 23.252688172043)
(3364.08, 23.0864459535937)
(3411.46, 22.923740803622)
(3458.59, 22.7681097906055)
(3505.44, 22.6195529145444)
(3552.55, 22.4603848330504)
(3599.57, 22.294142614601)
(3646.96, 22.134974533107)
(3694.45, 21.9864176570458)
(3755.82, 21.7812676853424)
(3948.96, 21.1870401810979)
(4141.84, 20.6140350877193)
(4333.04, 20.0693265421619)
(4525.45, 19.5316921335597)
(4719.11, 19.0223542727787)
(4916.36, 18.5483870967742)
(5114.35, 18.0850311262026)
(5315.69, 17.6393604980192)
(5518.43, 17.1972269383135)
(5722.29, 16.783389926429)
(5922.7, 16.383701188455)
(6122.17, 15.9875495189587)
(6321.72, 15.6020090548953)
(6520.08, 15.2306168647425)
(6716.53, 14.8698358800226)
(6910.53, 14.5090548953028)
(7107.47, 14.1518109790606)
(7305.34, 13.7910299943407)
(7500.93, 13.4408602150538)
(7695.44, 13.0942275042445)
(7890.03, 12.7475947934352)
(8079.76, 12.4080362195812)
(8273.23, 12.0684776457272)
(8464.4, 11.7218449349179)
(8649.33, 11.371675155631)
(8831.08, 11.0108941709112)
(9011.09, 10.6501131861913)
(9185.65, 10.2928692699491)
(9355.49, 9.91794001131862)
(9518.82, 9.54301075268817)
(9676.54, 9.16100735710243)
(9829.87, 8.76839275608376)
(9978.74, 8.37224108658744)
(10123.79, 7.97608941709111)
(10267.19, 7.56225240520657)
(10409.34, 7.10950764006791)
(10553.14, 6.57894736842105)
(10700.37, 5.89275608375778)
(10849.62, 3.74221844934918)
};
\addplot[black,very thick,mark=*] coordinates {
(0.23, -0.686191284663271)
(457.39, -0.431522354272779)
(966.49, 0)
(10849.62, 0)
};
\end{semilogxaxis}
\end{tikzpicture}
\caption{The progress of the MIP algorithm for a Random instance of size $35 \times 50$.  The gap between the proven upper bound and the optimal solution remains significant until the very last steps of the algorithm.  The starting solution is only 0.68\% away from the optimal one, and the optimal solution is obtained after 966~seconds.}
\label{fig:mip-dynamics}
\end{figure*}

Let us now discuss experimental results using our heuristics. Since there are several basic heuristic approaches that can be combined to form compound heuristics, it is important to adapt appropriate naming conventions to indicate the nature of the heuristics (combinations) considered. Where appropriate, we denote the algorithms as follows.  `Rn', `G', `A' and `F' stand for Random (a procedure that generates a random solution), Greedy, Alternating and Flip algorithms, respectively.  We use parentheses to indicate the initial solution for a local search.  For example, `A(G)' denotes the Alternating local search applied to a Greedy solution.  `P$_k$', `V$_k$' and V$^\text{ex}_k$ stand for Random Portions, VND and VND Exhaustive algorithms, respectively.  Multi-Start metaheuristic is denoted by `M', where the local search procedure is specified in parentheses.  The Clustering Row-Merge, Multi-Start Row-Merge and Row-Merge Local Search algorithms are denoted by `R$_k$', `R$^\text{m}_k$' and `R$^\text{ls}_k$', respectively. For quick reference, Table~\ref{hnames} provides a summary of notations for some of the important heuristics we discuss here.
\begin{table*}[ht] \centering
\begin{tabular}{@{}lp{15.0cm}@{}}
\toprule
Rn & Construction algorithm producing random solutions.\\
G & Greedy construction algorithm.\\
F & Flip local search algorithm.\\
A & Alternating local search algorithm.\\
$V^{ex}_1$ & VND exhaustive algorithm combining Alternating and Flip local searches.\\
M & Multi-Start metaheuristic.  In most of the experiments, terminates after a certain amount of time.\\
P$_k$ & Portions local search; $k$ is the `portion size'.  Terminates after a certain amount of time.\\
R$_k$ & Clustering Row-Merge construction heuristic; $k$ is the number of clusters.\\
R$^\text{m}_k$ & Multi-Start Row-Merge construction heuristic; $k$ is the number of clusters used in each iteration.\\
R$^\text{ls}_k$ & Row-Merge local search heuristic; $k$ is the number of clusters.\\
\bottomrule
\end{tabular}
\caption{Naming convention of the algorithms.}
\label{hnames}
\end{table*}

\subsection{Experiments with Small instances}
We initially tested the heuristic algorithms on the small instances.  Our experiments showed that the performance of the fastest algorithms (Greedy, Alternating, Flip and VND Exhaustive ($k = 1$)) is relatively good, on average, but those algorithms rarely reach optimality.  Increasing the value of $k$ in the VND Exhaustive heuristic significantly improves the solution quality.  However, even when $k = 7$, VND Exhaustive was not able to obtain optimal solutions for certain instances while the running time of the heuristic was roughly 10~minutes for each of the instances of size $50 \times 50$.

It turns out that one of the most efficient approaches for the small instances is the Multi-Start heuristic used with a fast local search.  Even being given only 10 iterations, M(V$^\text{ex}_1$) produced optimal solutions for almost all the Random instances.\footnote{In fact, optimal solutions to some of the test problem were not available to us (see Table~\ref{tab:compare-exact}), in which cases we used the best known solutions to evaluate the heuristics.}  With 100 iterations, M(V$^\text{ex}_1$) reaches optimality for all the Random, Max Induced Subgraph, MaxCut and Matrix Factorization instances.  By further increasing the number of iterations, we were able to improve the quality of M(V$^\text{ex}_1$) solutions for the Max Biclique problems though, even for $100\,000$ iterations, the algorithm did not reach the optimal solution for the Max Biclique problem of size $50 \times 50$.  Note, however, that for most of the test instances, M(V$^\text{ex}_1$) ($100\,000$ iterations) terminated in less than 10~seconds.  Moreover, the performance of M(V$^\text{ex}_1$) with only 1000 iterations was almost as good as its performance with $100\,000$ iterations.

The Row-Merge heuristics also work well for small instances.  For example, R$^\text{m}_{20}$, where the reduced problem is solved using V$^\text{ex}_2$  with reasonable time restriction (less than 1~second for most of the instances), yielded optimal solutions for all the test problems.

Using the insight obtained from experimental results on small size instances, we have decided on initial parameter setting and examined the scalability of the algorithms using medium and large instances.

\subsection{Experiments with Fast Heuristics}
\label{sec:evaluation-fast-heuristics}

We first analyse the performance of our fastest heuristics: Random, Greedy, Alternating, Flip and VND Exhaustive ($k = 1$).  The results of the experiments for the Random and Max Biclique instances can be found in Tables~\ref{tab:fast-random} and \ref{tab:fast-random-biclique}.  For each of the heuristics, the tables reports the running time and the relative gap between the solution obtained and the best known solution calculated as
\begin{equation}
\label{eq:gap-definition}
\mathit{gap}(f) = \frac{f_\text{best} - f}{f_\text{best}} \cdot 100 \% \,,
\end{equation}
where $f$ is the objective value of the solution obtained and $f_\text{best}$ is the objective of the best known solution for that test instance (note that for any practical instance $f_\text{best} > 0 $).  Where uncertainty is present (e.g., in `Rn' or `A(Rn)'), the experiment is repeated 10 times, and the average value is reported.  The last row of the table reports the average objective gap and the average running time for each heuristic.  These values, however, should be considered with caution.  For example, the running time of a heuristic may vary dramatically for different instances and, thus, the average running time almost does not reflect the performance of the algorithm for smaller instances.

\begin{table*}[ht] \centering
\begin{tabular}{@{} r @{} c @{} r r r r r r r r @{} c @{} r r r r @{}}
\toprule
&&\multicolumn{8}{c}{Gap to the best known, \%}&&\multicolumn{4}{c}{Time, ms}\\
\cmidrule(){3-10}
\cmidrule(){12-15}
$m \times n \phantom{000}$&\hspace*{2em}&Rn&G&A(Rn)&A(G)&F(Rn)&F(G)&V$^\text{ex}_1$(Rn)&V$^\text{ex}_1$(G)&\hspace*{2em}&G&A(G)&F(G)&V$^\text{ex}_1$(G)\\
\midrule
$\phantom{0}100 \times 1000$&&97.8&0.7&2.8&0.7&\underline{0.1}&0.2&0.1&0.2&&\underline{2.1}&2.7&3.4&24.5\\
$\phantom{0}200 \times 1000$&&103.4&1.8&2.9&0.3&0.3&\underline{0.1}&0.4&0.2&&\underline{12.0}&14.2&15.7&37.4\\
$\phantom{0}400 \times 1000$&&101.1&3.0&1.9&1.6&\underline{0.5}&0.9&0.6&0.5&&\underline{7.5}&11.9&14.7&54.6\\
$\phantom{0}600 \times 1000$&&98.5&4.7&1.3&1.4&\underline{0.6}&1.0&0.7&0.8&&\underline{14.2}&21.0&31.8&77.3\\
$\phantom{0}800 \times 1000$&&100.7&4.1&1.2&0.9&\underline{0.5}&0.6&0.7&0.8&&\underline{12.8}&19.3&45.6&80.7\\
$1000 \times 1000$&&102.9&4.0&1.2&1.2&1.0&\underline{0.5}&1.0&1.2&&\underline{20.2}&31.2&79.4&109.9\\[1ex]
$\phantom{0}500 \times 5000$&&99.5&1.1&3.0&1.0&0.2&0.1&0.2&\underline{0.1}&&\underline{42.5}&71.5&188.3&471.8\\
$1000 \times 5000$&&99.2&2.6&1.7&1.3&0.3&0.3&0.3&\underline{0.2}&&\underline{115.5}&175.9&421.3&926.4\\
$2000 \times 5000$&&99.2&3.4&1.2&0.9&\underline{0.4}&0.5&0.5&0.4&&\underline{180.8}&337.3&976.9&2940.1\\
$3000 \times 5000$&&101.3&3.7&1.1&0.9&0.6&\underline{0.3}&0.7&0.5&&\underline{289.9}&524.2&2596.9&4195.6\\
$4000 \times 5000$&&100.5&4.1&0.9&0.9&0.5&0.6&0.7&\underline{0.4}&&\underline{367.7}&652.0&1721.3&7538.2\\
$5000 \times 5000$&&99.5&4.4&1.0&1.0&\underline{0.8}&0.9&0.8&1.0&&\underline{407.9}&774.7&4806.9&1750.1\\
\midrule
Average&&100.3&3.1&1.7&1.0&\underline{0.5}&0.5&0.6&0.5&&\underline{122.8}&219.6&908.5&1517.2\\
\bottomrule
\end{tabular}
\caption{Evaluation of fast heuristics on the Random instances.  The running times of the local searches started from random solutions are close to those started from Greedy solutions.  The time needed to generate a random solution is negligible and, thus, is also skipped.}
\label{tab:fast-random}
\end{table*}

\begin{table*}[ht] \centering
\setlength{\tabcolsep}{0.5em}
\begin{tabular}{@{} r @{} c @{} r r r r r r r r @{} c @{} r r r r @{}}
\toprule
&&\multicolumn{8}{c}{Gap to the best known, \%}&&\multicolumn{4}{c}{Time, ms}\\
\cmidrule(){3-10}
\cmidrule(){12-15}
$m \times n \phantom{000}$&\hspace*{2em}&Rn&G&A(Rn)&A(G)&F(Rn)&F(G)&V$^\text{ex}_1$(Rn)&V$^\text{ex}_1$(G)&\hspace*{2em}&G&A(G)&F(G)&V$^\text{ex}_1$(G)\\
\midrule
$\phantom{0}100 \times 1000$&&417961.3&1.9&91.6&1.9&\underline{0.0}&\underline{0.0}&\underline{0.0}&\underline{0.0}&&\underline{20.1}&20.6&20.9&24.0\\
$\phantom{0}200 \times 1000$&&499709.4&5.5&100.0&5.5&18.9&2.2&\underline{1.3}&2.2&&\underline{13.4}&14.5&17.6&67.3\\
$\phantom{0}400 \times 1000$&&560477.3&18.0&100.0&13.2&91.3&7.5&\underline{3.2}&7.5&&\underline{6.9}&10.2&17.4&48.6\\
$\phantom{0}600 \times 1000$&&590077.1&22.8&97.4&21.0&92.6&8.3&\underline{6.2}&8.3&&\underline{56.5}&61.4&76.2&109.9\\
$\phantom{0}800 \times 1000$&&647673.5&33.1&99.1&30.4&89.5&11.4&\underline{6.7}&13.4&&\underline{13.7}&19.8&44.1&91.0\\
$1000 \times 1000$&&633927.8&47.7&97.7&45.5&85.7&27.2&\underline{13.7}&30.3&&\underline{15.7}&23.7&50.5&137.9\\[1ex]
$\phantom{0}500 \times 5000$&&3239937.6&4.0&100.0&3.6&4.7&0.0&\underline{0.0}&0.0&&\underline{44.8}&67.8&104.1&248.9\\
$1000 \times 5000$&&3924155.1&15.1&100.0&14.6&76.9&2.9&\underline{0.5}&5.1&&\underline{85.8}&127.7&262.2&638.3\\
$2000 \times 5000$&&4506581.7&32.6&99.8&30.8&95.6&13.1&\underline{12.1}&13.1&&\underline{164.6}&241.2&721.6&1666.6\\
$3000 \times 5000$&&4841943.6&49.4&99.4&47.1&96.4&28.5&\underline{14.2}&28.1&&\underline{236.0}&345.9&950.5&1744.6\\
$4000 \times 5000$&&5155697.4&48.4&99.9&33.1&98.9&32.2&\underline{12.1}&31.9&&\underline{322.7}&478.0&829.1&1448.2\\
$5000 \times 5000$&&4918821.7&66.3&99.8&65.5&98.9&45.0&\underline{19.0}&44.6&&\underline{379.8}&584.8&1156.9&3211.4\\
\midrule
Average&&2494747.0&28.7&98.7&26.0&70.8&14.9&\underline{7.4}&15.4&&\underline{113.3}&166.3&354.3&786.4\\
\bottomrule
\end{tabular}
\caption{Evaluation of fast heuristics on the Max Biclique instances.  The running times of the local searches started from random solutions are close to those started from Greedy solutions.  The time needed to generate a random solution is negligible and, thus, is also skipped.}
\label{tab:fast-random-biclique}
\end{table*}

Recall that the expected objective value of a random solution can be calculated using equation (\ref{eq:expected-random-solution}).  Since the expected weights $\bar{q}$, $\bar{c}$ and $\bar{d}$ in these instances are zeros, the expected objective value of solutions produced by the Random heuristic is also zero and the objective gaps are close to 100\%, see~(\ref{eq:gap-definition}).  This result also holds for the Max Induced Subgraph, MaxCut and Matrix Factorization problems.  However, it is different for a Max Biclique instance for which the expected objective value is a large negative number (recall that, in a Max Biclique problem, $q_{ij} = -M$ if there is no edge between $i$ and $j$ in the original bipartite graph).  Thus, the objective gaps of the random solutions for Max Biclique are very large, and the Alternating algorithm applied to a random solution quickly converges to near-zero solutions (recall that, according to Theorem~\ref{th:alternate}, the solution improved with the Alternating local search can never have negative objective).  It shows that the Alternating algorithm should not be normally applied to solutions with negative objective values as the Alternating local search tends to fall into a trivial local maximum of $(x, y) = (\zerovector{m}, \zerovector{n})$ in such cases.  Also note that, among the quick heuristics, V$^\text{ex}_1$(Rn) shows the best performance for Max Biclique problems and F(Rn) performs poorly.

Other than that, the performance of the fast heuristics has the same pattern for the other instance types.  The Greedy algorithm is a good option to obtain a reasonable solution in a very short time; it is the fastest heuristic in most of the experiments.  The Alternating and Flip algorithms are yet two fast and efficient local searches; the Alternating heuristic performs faster but the Flip algorithm yields better solutions.  The solution quality of V$^\text{ex}_1$ is similar to that of the Flip heuristic though V$^\text{ex}_1$ is notably slower.  However, the neighborhood of V$^\text{ex}_1$ is clearly larger than that of the Flip local search and, thus, being applied to a near-optimal solution, V$^\text{ex}_1$ is expected to demonstrate superior performance.

Whether we start a local search from a random or a Greedy solution, the performance of the local searches remains approximately the same.  However, we noticed that starting from the Greedy algorithm yields, on average, slightly better solutions and, thus, in what follows, we have used Greedy solutions as the starting ones by default.

\subsection{Experiments with Portions-Based and Multi-Start Heuristics}
\label{sec:evaluation-slower-heuristics}

In Tables~\ref{tab:slow-random} and \ref{tab:slow-biclique}, we compare Portions-based and Multi-Start heuristics, which include the VND Exhaustive (V$^\text{ex}_k$), VND (V$_k$), Random Portions (P$_k$) and Multi-Start (M(A), M(F) and M(V$^\text{ex}_1$)) algorithms.  For a fair competition, we give each of the Random Portions and Multi-Start algorithms the same amount of time as the VND ($k = 6$) takes to solve the corresponding instance.  We also studied performance of the heuristics being given smaller and larger times but this did not yield any new observations.

\begin{table*}[ht] \centering
\begin{tabular}{@{} r @{} c @{} r r r r r r r r r r @{} c @{} r r r @{}}
\toprule
&&\multicolumn{10}{c}{Gap to the best known, \%}&&\multicolumn{3}{c}{Time, sec}\\
\cmidrule(){3-12}
\cmidrule(){14-16}
$m \times n \phantom{000}$&\hspace*{2em}&V$^\text{ex}_1$&V$^\text{ex}_2$&V$_{6}$&M(A)&M(F)&M(V$^\text{ex}_1$)&P$_3$&P$_4$&P$_6$&P$_8$&\hspace*{2em}&V$^\text{ex}_1$&V$^\text{ex}_2$&V$_{6}$\\
\midrule
$\phantom{0}100 \times 1000$&&0.18&0.04&0.08&0.77&\underline{0.00}&\underline{0.00}&0.06&0.05&0.09&0.10&&\underline{0.0}&0.4&0.2\\
$\phantom{0}200 \times 1000$&&0.15&0.10&0.11&0.80&\underline{0.00}&0.01&0.11&0.10&0.11&0.15&&\underline{0.0}&1.5&0.5\\
$\phantom{0}400 \times 1000$&&0.52&0.52&0.52&0.59&0.09&\underline{0.08}&0.41&0.40&0.37&0.68&&\underline{0.1}&2.5&0.7\\
$\phantom{0}600 \times 1000$&&0.79&0.57&0.74&0.31&\underline{0.11}&0.13&0.71&0.63&0.71&0.70&&\underline{0.1}&18.5&1.2\\
$\phantom{0}800 \times 1000$&&0.77&0.62&0.74&0.46&\underline{0.25}&0.30&0.54&0.52&0.62&0.54&&\underline{0.1}&37.8&1.4\\
$1000 \times 1000$&&1.15&0.98&1.14&0.48&\underline{0.35}&0.39&0.53&0.56&0.56&0.59&&\underline{0.1}&54.3&2.1\\[1ex]
$\phantom{0}500 \times 5000$&&0.05&\underline{0.02}&0.04&1.88&0.05&0.03&0.16&0.15&0.12&0.11&&\underline{0.5}&66.0&4.7\\
$1000 \times 5000$&&0.21&0.19&0.20&1.09&\underline{0.15}&0.16&0.28&0.30&0.27&0.30&&\underline{0.9}&257.2&10.3\\
$2000 \times 5000$&&0.49&0.49&0.49&0.83&\underline{0.30}&0.33&0.44&0.44&0.47&0.44&&\underline{2.9}&1\,243.5&20.4\\
$3000 \times 5000$&&0.55&0.55&0.55&0.69&\underline{0.36}&0.42&0.50&0.49&0.55&0.46&&\underline{4.2}&686.1&34.5\\
$4000 \times 5000$&&0.44&0.40&0.44&0.55&\underline{0.39}&0.45&0.57&0.56&0.58&0.53&&\underline{7.5}&5\,490.2&48.5\\
$5000 \times 5000$&&0.97&0.62&0.97&0.68&0.57&\underline{0.57}&0.73&0.80&0.74&0.72&&\underline{1.8}&29\,088.8&42.5\\
\midrule
Average&&0.52&0.43&0.50&0.76&\underline{0.22}&0.24&0.42&0.42&0.43&0.44&&\underline{1.5}&3\,078.9&13.9\\
\bottomrule
\end{tabular}
\caption{Evaluation of slow heuristics on the Random instances.  The running times of M(A), M(F), M(V$^\text{ex}_1$), P$_3$, P$_4$, P$_6$ and P$_8$ are equal to that of V$_6$.}
\label{tab:slow-random}
\end{table*}

\begin{table*}[ht] \centering
\begin{tabular}{@{} r @{} c @{} r r r r r r r r r r @{} c @{} r r r @{}}
\toprule
&&\multicolumn{10}{c}{Gap to the best known, \%}&&\multicolumn{3}{c}{Time, sec}\\
\cmidrule(){3-12}
\cmidrule(){14-16}
$m \times n\phantom{000}$&\hspace*{2em}&V$^\text{ex}_1$&V$^\text{ex}_2$&V$_{6}$&M(A)&M(F)&M(V$^\text{ex}_1$)&P$_3$&P$_4$&P$_6$&P$_8$&\hspace*{2em}&V$^\text{ex}_1$&V$^\text{ex}_2$&V$_{6}$\\
\midrule
$\phantom{0}100 \times 1000$&&\underline{0.00}&\underline{0.00}&\underline{0.00}&87.95&\underline{0.00}&\underline{0.00}&\underline{0.00}&\underline{0.00}&\underline{0.00}&\underline{0.00}&&\underline{0.0}&0.2&0.2\\
$\phantom{0}200 \times 1000$&&2.21&\underline{0.00}&1.49&94.76&2.74&0.13&1.07&0.60&1.08&1.18&&\underline{0.1}&9.4&0.3\\
$\phantom{0}400 \times 1000$&&7.51&7.41&7.51&96.81&75.05&\underline{3.21}&7.53&7.51&7.70&7.52&&\underline{0.0}&6.6&0.6\\
$\phantom{0}600 \times 1000$&&8.33&\underline{1.14}&5.61&92.59&85.64&2.38&6.10&6.03&6.00&7.07&&\underline{0.1}&24.2&1.7\\
$\phantom{0}800 \times 1000$&&13.45&\underline{0.49}&11.46&95.29&85.40&6.01&11.02&11.75&11.76&12.24&&\underline{0.1}&51.2&1.5\\
$1000 \times 1000$&&30.35&11.11&25.64&95.41&78.20&\underline{10.88}&25.33&24.60&25.54&26.32&&\underline{0.1}&85.7&2.6\\[1ex]
$\phantom{0}500 \times 5000$&&0.03&0.03&0.03&98.58&4.72&\underline{0.00}&0.03&0.03&0.03&0.03&&\underline{0.2}&16.8&3.6\\
$1000 \times 5000$&&5.14&\underline{0.20}&2.33&98.56&1.09&0.40&1.37&2.03&1.45&3.04&&\underline{0.6}&219.5&11.6\\
$2000 \times 5000$&&13.12&\underline{6.23}&12.71&98.47&93.27&12.06&12.82&13.08&13.37&14.06&&\underline{1.7}&1\,684.6&20.9\\
$3000 \times 5000$&&28.11&\underline{4.74}&27.36&98.30&94.81&8.13&27.20&27.53&27.58&27.97&&\underline{1.7}&5\,205.6&43.3\\
$4000 \times 5000$&&31.87&12.70&27.09&98.98&98.02&\underline{8.92}&28.29&27.95&29.83&30.80&&\underline{1.4}&5\,866.5&147.3\\
$5000 \times 5000$&&44.62&21.47&43.14&99.12&97.64&\underline{15.45}&41.85&42.26&42.58&43.89&&\underline{3.2}&9\,012.8&121.2\\
\midrule
Average&&15.39&\underline{5.46}&13.70&96.24&59.71&5.63&13.55&13.61&13.91&14.51&&\underline{0.8}&1\,848.6&29.6\\
\bottomrule
\end{tabular}
\caption{Evaluation of slow heuristics on the Max Biclique instances.  The running times of M(A), M(F), M(V$^\text{ex}_1$), P$_3$, P$_4$, P$_6$ and P$_8$ are equal to that of V$_6$.}
\label{tab:slow-biclique}
\end{table*}

The most successful heuristic in this series of experiments is Multi-Start based on either Flip or V$^\text{ex}_1$ local search.  Usually, M(F) dominates M(V$^\text{ex}_1$); indeed, according to Table~\ref{tab:fast-random}, Flip is significantly faster than V$^\text{ex}_1$ and, thus, Multi-Start is able to complete more iterations when using the Flip improvement procedure.  However, for the Max Biclique instances, M(F) performs very poorly (see Table~\ref{tab:slow-biclique}) as a result of F(Rn) low performance (see Section~\ref{sec:evaluation-fast-heuristics}) and, therefore, M(V$^\text{ex}_1$) is preferable for the instances like Max Biclique.

It is worth noting that for the Matrix Factorization instances, the Multi-Start algorithms are sometimes dominated by the P$_3$ and/or P$_4$ heuristics although, on average, M(F) shows slightly better performance; a possible reason for that is discussed in Section~\ref{sec:evaluation-rowmerge-heuristics}.  In several experiments, M(F) and M(V$^\text{ex}_1$) were also outperformed by V$^\text{ex}_2$ but at the cost of unreasonable running times.  We recommend to use V$^\text{ex}_k$ with $k > 1$ only for relatively small instances.

According to our experiments, the optimal value of $k$ for the Random Portions heuristic is $k = 4$ (sometimes $k = 3$).  We conducted the tests for a wide range of $k$, given times and instances, and P$_4$ was the winner in most such experiments.  We believe, however, that, being given better starting solutions, the Random Portions heuristic would benefit from a larger value of $k$.

Performance of the VND ($k = 6$) is slightly worse, on average, than that of P$_4$, and this tendency holds for all the instance types.  In our experiments, we also tried a modification of the VND heuristic that restarts immediately after finding an improvement.  Considering that the P$_k$ algorithm is usually significantly slower than the Alternating and Flip local searches, this could speed up the heuristic.  However, our experiments showed that such a change did not improve the performance of the VND.

\subsection{Experiments with Row-Merge Heuristics}
\label{sec:evaluation-rowmerge-heuristics}

Unlike all other algorithms discussed above, the Clustering Row-Merge heuristic requires a number of good solutions to perform row clusterization.  In an evolutionary algorithm, one can, e.g., use a subset of the last generation for this purpose.  In iterative heuristics, one can use the best $p$ distinct solutions obtained so far.  In our experiments, we generate synthetic solutions by producing $p = 100$ random solutions which are then further improved with VND Exhaustive ($k = 1$).  Note that we do not include the time spent on generating these solutions in the running time of the Clustering Row-Merge heuristic.

Recall that the Clustering Row-Merge algorithm solves a constrained version of the BQP\@.  Thus, even if the obtained solution is optimal with respect to the BQP$(Q, c, d, \partition)$, it is unlikely that it remains a local maximum for either of the neighborhoods defined above after relaxation of the constraints (\ref{eq:bqp-partition-constraint}).  In our experiments, we apply V$^\text{ex}_1$ to every solution obtained by the Clustering Row-Merge heuristic.

As it was mentioned in Section~\ref{sec:rowmerge}, the Row-Merge Local Search cannot be efficient if $k$ is small.  In particular, it is necessary to keep the clusters in $\partition$ small.  Thus, we use $k = \lfloor m / 2 \rfloor$, $k = \lfloor m / 3 \rfloor$, etc.  The resulting reduced BQP instances are, therefore, too large to be solved exactly.  We use F(G) to obtain a solution $(x^{*}, y^{*})$ of such an instance.  Same approach is used within the Multi-Start Row-Merge heuristic as Multi-Start Row-Merge also performs better for small clusters (indeed, if a random cluster is large, then there is a high chance that neither assigning 0 nor 1 to the corresponding variable $x_{i^*}$ is good).  These decisions were made after extensive experimentation with different construction and improvement procedures.

\begin{table*}[ht] \centering
\begin{tabular}{@{} r @{} c @{} r r r r r r r r r @{} c @{} r r r r @{}}
\toprule
&&\multicolumn{9}{c}{Gap to the best known, \%}&&\multicolumn{4}{c}{Time, sec}\\
\cmidrule(){3-11}
\cmidrule(){13-16}
$m \times n \phantom{000}$&\hspace*{2em}&V$_6$&M$_{100}$&M(F)&M(V$^\text{ex}_1$)&R$_{15}$&R$_{20}$&R$^\text{m}_{\lfloor m / 2\rfloor}$&R$^\text{m}_{\lfloor m / 3\rfloor}$&R$^\text{ls}_{\lfloor m / 2\rfloor}$&\hspace*{2em}&V$_6$&M$_{100}$&R$_{15}$&R$_{20}$\\
\midrule
$100 \times 1000$&&0.08&\underline{0.00}&\underline{0.00}&\underline{0.00}&\underline{0.00}&\underline{0.00}&0.00&0.00&0.03&&\underline{0.2}&0.5&0.5&13.9\\
$200 \times 1000$&&0.11&\underline{0.00}&0.00&0.01&0.08&0.03&0.00&0.00&0.10&&0.5&1.1&\underline{0.4}&13.8\\
$400 \times 1000$&&0.52&\underline{0.06}&0.09&0.08&0.15&0.14&0.11&0.11&0.23&&0.7&2.2&\underline{0.6}&13.4\\
$600 \times 1000$&&0.74&\underline{0.10}&0.11&0.13&0.23&0.22&0.11&0.17&0.39&&1.2&3.3&\underline{0.6}&14.1\\
$800 \times 1000$&&0.74&0.23&0.25&0.30&0.28&0.26&0.21&\underline{0.20}&0.31&&1.4&4.0&\underline{0.9}&14.7\\
$1000 \times 1000$&&1.14&\underline{0.29}&0.35&0.39&0.35&0.31&0.42&0.34&0.33&&2.1&5.8&\underline{1.0}&15.9\\[1ex]
$500 \times 5000$&&0.04&\underline{0.01}&0.05&0.03&0.03&0.02&0.04&0.05&0.07&&4.7&22.5&\underline{2.5}&71.1\\
$1000 \times 5000$&&0.20&\underline{0.09}&0.15&0.16&0.12&0.11&0.14&0.13&0.18&&10.3&52.9&\underline{3.6}&71.7\\
$2000 \times 5000$&&0.49&\underline{0.25}&0.30&0.33&0.28&0.32&0.30&0.30&0.31&&20.4&131.6&\underline{7.3}&73.2\\
$3000 \times 5000$&&0.55&0.35&0.36&0.42&0.33&\underline{0.31}&0.43&0.39&0.40&&34.5&214.3&\underline{12.3}&79.9\\
$4000 \times 5000$&&0.44&0.35&0.39&0.45&0.33&\underline{0.30}&0.39&0.39&0.43&&48.5&286.9&\underline{17.2}&94.7\\
$5000 \times 5000$&&0.97&0.49&0.57&0.57&\underline{0.33}&0.35&0.61&0.56&0.63&&42.5&348.5&\underline{40.4}&111.1\\
\midrule
Average&&0.50&\underline{0.19}&0.22&0.24&0.21&0.20&0.23&0.22&0.28&&13.9&89.5&\underline{7.3}&49.0\\
\bottomrule
\end{tabular}
\caption{Evaluation of the Row-Merge heuristics on the Random instances.  The running time of M(F), M(V$^\text{ex}_1$), R$^\text{m}_{\lfloor m / 2 \rfloor}$, R$^\text{m}_{\lfloor m / 3 \rfloor}$ and R$^\text{ls}_{\lfloor m / 2 \rfloor}$ is fixed to that of V$_6$.  The algorithm M$_{100}$ denotes M(V$^\text{ex}_1$) terminated after 100 iterations.}
\label{tab:rowmerge}
\end{table*}

\begin{table*}[ht] \centering
\setlength{\tabcolsep}{0.3em}
\begin{tabular}{@{} r @{} c @{} r r r r r r r r r @{} c @{} r r r r @{}}
\toprule
&&\multicolumn{9}{c}{Gap to the best known, \%}&&\multicolumn{4}{c}{Time, sec}\\
\cmidrule(){3-11}
\cmidrule(){13-16}
$m \times n \phantom{000}$&\hspace*{2em}&V$_6$&M$_{100}$&M(F)&M(V$^\text{ex}_1$)&R$_{15}$&R$_{20}$&R$^\text{m}_{\lfloor m / 2\rfloor}$&R$^\text{m}_{\lfloor m / 3\rfloor}$&R$^\text{ls}_{\lfloor m / 2\rfloor}$&\hspace*{2em}&V$_6$&M$_{100}$&R$_{15}$&R$_{20}$\\
\midrule
$100 \times 1000$&&\underline{0.00}&\underline{0.00}&\underline{0.00}&\underline{0.00}&\underline{0.00}&\underline{0.00}&\underline{0.00}&\underline{0.00}&\underline{0.00}&&\underline{0.2}&0.5&0.4&10.8\\
$200 \times 1000$&&1.49&\underline{0.00}&2.74&0.13&\underline{0.00}&\underline{0.00}&\underline{0.00}&\underline{0.00}&0.50&&\underline{0.3}&0.7&0.4&12.1\\
$400 \times 1000$&&7.51&3.21&75.05&3.21&3.21&3.21&0.35&\underline{0.03}&6.60&&0.6&2.0&\underline{0.5}&12.5\\
$600 \times 1000$&&5.61&2.38&85.64&2.38&2.38&2.38&0.92&\underline{0.61}&2.48&&1.7&3.0&\underline{0.6}&11.9\\
$800 \times 1000$&&11.46&6.01&85.40&6.01&2.93&2.93&1.70&\underline{1.20}&4.40&&1.5&3.9&\underline{0.7}&12.2\\
$1000 \times 1000$&&25.64&10.88&78.20&10.88&7.79&7.79&\underline{1.10}&1.65&20.09&&2.6&4.2&\underline{1.0}&12.2\\[1ex]
$500 \times 5000$&&0.03&\underline{0.00}&4.72&\underline{0.00}&\underline{0.00}&\underline{0.00}&\underline{0.00}&\underline{0.00}&0.03&&3.6&14.1&\underline{2.2}&61.1\\
$1000 \times 5000$&&2.33&0.36&1.09&0.40&0.31&0.31&\underline{0.23}&0.46&0.28&&11.6&24.0&\underline{2.6}&58.2\\
$2000 \times 5000$&&12.71&12.06&93.27&12.06&12.07&12.07&1.14&\underline{0.71}&10.19&&20.9&61.0&\underline{5.7}&74.3\\
$3000 \times 5000$&&27.36&8.13&94.81&8.13&8.13&8.13&1.58&\underline{0.90}&10.91&&43.3&79.9&\underline{10.4}&84.4\\
$4000 \times 5000$&&27.09&8.92&98.02&8.92&26.63&26.63&1.85&\underline{0.81}&6.51&&147.3&92.2&\underline{15.6}&82.2\\
$5000 \times 5000$&&43.14&15.45&97.64&15.45&15.45&15.45&4.34&\underline{3.24}&10.08&&121.2&122.4&\underline{36.1}&125.5\\
\midrule
Average&&13.70&5.62&59.71&5.63&6.58&6.58&1.10&\underline{0.80}&6.01&&29.6&34.0&\underline{6.3}&46.4\\
\bottomrule
\end{tabular}
\caption{Evaluation of the Row-Merge heuristics on the Biclique instances.  The running time of M(F), M(V$^\text{ex}_1$), R$^\text{m}_{\lfloor m / 2 \rfloor}$, R$^\text{m}_{\lfloor m / 3 \rfloor}$ and R$^\text{ls}_{\lfloor m / 2 \rfloor}$ is fixed to that of V$_6$.  The algorithm M$_{100}$ denotes M(V$^\text{ex}_1$) terminated after 100 iterations.}
\label{tab:rowmerge-biclique}
\end{table*}

The results of computational experiments for the most successful variations of the Row-Merge heuristics are reported in Tables~\ref{tab:rowmerge} and \ref{tab:rowmerge-biclique}.  The Clustering Row-Merge heuristic shows outstanding performance for the Random instances and performs similarly to M(F) and M(V$^\text{ex}_1$) for the other instances.  Note, however, that generation of initial solutions (the corresponding procedure is denoted as M$_{100}$ in Tables~\ref{tab:rowmerge} and \ref{tab:rowmerge-biclique}) for the Clustering Row-Merge heuristic often takes more time than running the Clustering Row-Merge algorithm itself and produces solutions superior to the solutions obtained by the Clustering Row-Merge heuristic.  Our experiments with different values of $p$ showed that reduction on the number of initial solutions worsens the quality of the Clustering Row-Merge method.  Hence, the Clustering Row-Merge algorithm is impractical as a standalone heuristic.  We also tried to use a heuristic solver for the BQP$(Q, c, d, \partition)$ problems in order to increase the value of $k$.  We noticed that the better is the heuristic solver and the larger is $k$, the higher is the quality of the resulting solutions.  However, the best performance was observed for the exact algorithm for BQP$(Q, c, d, \partition)$ and $k = 15$ or $k = 20$.  It is also worth mentioning that we expect that using better initial solutions for the Clustering Row-Merge heuristic would improve its performance.

The Multi-Start Row-Merge and Row-Merge Local Search algorithms demonstrate very good performance for certain instances.  In particular, the Multi-Start Row-Merge and Row-Merge Local Search heuristics outperform, on average, all other heuristics for the Biclique and Matrix Factorization instances, respectively.  Indeed, iterative improvements may not be efficient for the Biclique problems characterized by deep local maxima.  The Multi-Start approach utilized in the Multi-Start Row-Merge algorithm is preferable in such a case.  In contrast, the landscape of the Matrix Factorization instances is rather flat and iterative improvements work well for it.  These results hold even if we allow more or less time to each of the heuristics.

\subsection{Transformation Based Solver}

Note that the BQP($Q,c,d$) can be formulated as a QP with $m+n$ variables:
\begin{align*}
\mbox{QP($\bar{Q},\bar{c}$):\hspace{.5cm}}&\mbox{Maximize } f(w) = w^T\bar{Q}w + \bar{c}w\\
&\mbox{subject to } w \in \{0,1\}^{m+n} \,,
\end{align*}
where
\begin{equation}\label{eq0}
\bar{Q}=\left[
\begin{array}{c|c}
O_{m\times m} & \frac{1}{2}Q\\ \hline
\frac{1}{2}Q & O_{n\times n}
\end{array}\right],\;
 \bar{c}\!=\!\left[\!
\begin{array}{c|c}
c & d
\end{array}\!\right]
 \mbox{ and }   w^T\!=\!\left[
\begin{array}{c|c}
\!x & y\!
\end{array}\right],
\end{equation}
and $O_{n \times n}$ and $O_{m \times m}$ are zero matrices.   Thus, any algorithm (exact or heuristic) to solve QP($\bar{Q},\bar{c}$) can be used to solve BQP as well.  This approach is particularly interesting because the QP problem is very well studied and there exist highly efficient algorithms for it (see, e.g., \citet{Lu2010,Wang2012}).  To assess the effectiveness of this approach, we compared some of our algorithms to a state-of-the-art heuristic for the QP~\citep{Wang2012}\@.  The parameters of the heuristic were determined empirically; we set the tabu tenure to 100 and the iterations cutoff to 10000.  We call the resulting BQP algorithm \emph{QP Based Solver}.

\begin{table*}[ht] \centering
\begin{tabular}{@{} r @{} c @{} r r r @{} c @{} r r r @{} c @{} r r r @{} c @{} r r r @{}}
\toprule
&&\multicolumn{3}{c}{Random}&&\multicolumn{3}{c}{Max Induced}&&\multicolumn{3}{c}{MaxCut}&&\multicolumn{3}{c}{Matrix Fact}\\
\cmidrule(){3-5}
\cmidrule(){7-9}
\cmidrule(){11-13}
\cmidrule(){15-17}
$m \times n \phantom{000}$&\hspace*{2em}&M(F)&R$^\text{m}_{\lfloor m / 3\rfloor}$&QP&\hspace*{2em}&M(F)&R$^\text{m}_{\lfloor m / 3\rfloor}$&QP&\hspace*{2em}&M(F)&R$^\text{m}_{\lfloor m / 3\rfloor}$&QP&\hspace*{2em}&M(F)&R$^\text{m}_{\lfloor m / 3\rfloor}$&QP\\
\midrule
$100 \times 1000$&&\underline{0.00}&0.00&0.61&&\underline{0.00}&\underline{0.00}&1.14&&\underline{0.71}&0.71&9.02&&1.33&1.56&\underline{0.24}\\
$200 \times 1000$&&\underline{0.00}&0.00&0.23&&\underline{0.03}&0.07&0.12&&1.55&\underline{1.51}&4.08&&0.64&0.69&\underline{0.04}\\
$400 \times 1000$&&\underline{0.09}&0.11&0.27&&\underline{0.04}&0.08&0.11&&3.25&3.23&\underline{2.33}&&0.40&0.48&\underline{0.14}\\
$600 \times 1000$&&\underline{0.11}&0.17&0.14&&0.27&0.27&\underline{0.21}&&3.52&4.22&\underline{1.24}&&0.31&0.32&\underline{0.12}\\
$800 \times 1000$&&0.25&0.20&\underline{0.08}&&0.35&0.37&\underline{0.11}&&4.52&4.29&\underline{1.81}&&0.83&0.80&\underline{0.12}\\
$1000 \times 1000$&&0.35&0.34&\underline{0.16}&&0.34&0.31&\underline{0.07}&&4.94&4.60&\underline{0.72}&&0.34&0.42&\underline{0.11}\\[1ex]
$500 \times 5000$&&\underline{0.05}&0.05&1.41&&\underline{0.02}&0.03&1.60&&\underline{1.38}&1.43&11.74&&1.56&1.44&\underline{0.60}\\
$1000 \times 5000$&&0.15&\underline{0.13}&0.46&&\underline{0.08}&0.10&0.60&&\underline{1.67}&1.79&7.73&&0.88&0.84&\underline{0.14}\\
$2000 \times 5000$&&0.30&0.30&\underline{0.15}&&0.23&0.27&\underline{0.14}&&2.70&\underline{2.56}&2.75&&0.58&0.62&\underline{0.04}\\
$3000 \times 5000$&&0.36&0.39&\underline{0.07}&&0.36&0.39&\underline{0.12}&&4.13&3.88&\underline{2.02}&&0.55&0.57&\underline{0.07}\\
$4000 \times 5000$&&0.39&0.39&\underline{0.05}&&0.54&0.53&\underline{0.16}&&4.91&4.79&\underline{1.95}&&0.62&0.58&\underline{0.08}\\
$5000 \times 5000$&&0.57&0.56&\underline{0.08}&&0.54&0.54&\underline{0.14}&&4.92&4.87&\underline{1.81}&&0.60&0.57&\underline{0.06}\\
\midrule
Average&&\underline{0.22}&0.22&0.31&&\underline{0.23}&0.25&0.38&&3.18&\underline{3.16}&3.93&&0.72&0.74&\underline{0.15}\\
\bottomrule
\end{tabular}
\caption{Evaluation of the QP Based Solver in comparison to M(F) and R$^\text{m}_{\lfloor m / 3\rfloor}$.  The running time of all the solvers is equal to that of V$_{6}$.}
\label{tab:qp}
\end{table*}

We report some of the results of our experiments in Table~\ref{tab:qp}.  An interesting observation is that the performance of the QP Based Solver significantly depends on the type of the instance.  For example, it clearly outperforms other algorithms for the Matrix Factorisation instances while losing the competition with Random instances.  Moreover, it might be impossible to apply the QP Based Solver to the Max Biclique problems using the above transformation\footnote{Recall that the Max Biclique instances contain large negative weights to guarantee that only existing edges of the underlying graph problem are included in optimal or near-optimal solutions.  As intermediate solutions may include several of such negative weights, their objectives often reach very large negative numbers.  In case of the QP solver we were using, such objective were causing overflows.}.

Another observation is that the QP Based Solver is performing very well for `square' or `near-square' instances (i.e., when $m \approx n$).  However, if $m \ll n$, the corresponding QP instance is significantly larger than the original BQP instance and this worsens the performance of the approach.

The more time the QP Based Solver is given, the better is its performance in comparison to the BQP specific algorithms.  Being given less time than in Table~\ref{tab:qp}, the QP Based Solver may not be able to produce any solution at all because its initialisation phase takes significant amount of time.  Such a behaviour is typical for a metaheuristic.

All this shows that the QP Based Solver is a viable approach though BQP specific algorithms are vital in certain circumstances.  Also note that we compare straightforward implementations of simple algorithms\footnote{Most noticeably, many of our algorithms could benefit from a more intelligent exact algorithm.} with an approach based on a sophisticated fine-tuned QP metaheuristic.

%

\section{Conclusion}
\label{sec:conclusion}

In this work, we considered an important combinatorial optimization problem called Bipartite Unconstrained 0-1 Programming Problem.  We proposed several heuristic approaches to construct and improve BQP solutions.  Several fast algorithms such as the Greedy construction heuristic and the Flip local search may be a reasonable choice if the solution quality requirements are low.

The success of slower techniques significantly depends on the instance type.  In particular, we noticed that simple Multi-Start heuristics perform very well for the Random instances.  For the instances with rather flat landscapes, iterative improvement methods work better.  For the problems having some `structure' (i.e., not purely random instances), the Row-Merge heuristics, being used with small clusters, provide superior solution quality comparing to simple Multi-Start algorithms.  Thus, we believe, this technique is of use in real applications.

The Clustering Row-Merge method also performs well for many problems though it can only be used within some metaheuristics.  A more sophisticated VND algorithm did not work well in our experiments.

Another viable approach to the BQP is a transformation of the problem into the QP and then solving it with well-developed QP algorithms.  This works particularly well for `square' instances (i.e., when $n \approx m$) and large running times.

In addition to the heuristics, we compared two rudimentary exact algorithms, namely the MIP approach and the Exhaustive algorithm.  It was shown that the former is generally faster for `square' instances, while the latter is significantly better when $n \gg m$.

Detailed complexity analysis of all algorithms and worst-case performance ratio of the greedy algorithm are given.  We believe that our study will motivate other researchers exploring more efficient algorithms for solving BQP\@.  Since the trade-off between computational time and solution quality are often contradictory in nature, no single algorithm or solution approach is likely to be considered as the best.  

In order to reduce the size of the manuscript, we excluded some of the tables from this paper.  The full version of our experimental results, test instances and best known solutions is available at \url{http://www.cs.nott.ac.uk/~dxk/}.

\bigskip

\noindent{\bf Acknowledgements:} We are thankful to Ying Wang for providing executable version of her code for solving QP\@.  Also, we express our appreciation to the anonymous referees for helpful suggestions which helped us to improve the paper.

\bibliographystyle{model2-names}
\bibliography{bqp}{}

\begin{thebibliography}{17}
\expandafter\ifx\csname natexlab\endcsname\relax\def\natexlab#1{#1}\fi
\providecommand{\url}[1]{\texttt{#1}}
\providecommand{\href}[2]{#2}
\providecommand{\path}[1]{#1}
\providecommand{\DOIprefix}{doi:}
\providecommand{\ArXivprefix}{arXiv:}
\providecommand{\URLprefix}{URL: }
\providecommand{\Pubmedprefix}{pmid:}
\providecommand{\doi}[1]{\href{http://dx.doi.org/#1}{\path{#1}}}
\providecommand{\Pubmed}[1]{\href{pmid:#1}{\path{#1}}}
\providecommand{\bibinfo}[2]{#2}
\ifx\xfnm\relax \def\xfnm[#1]{\unskip,\space#1}\fi
\bibitem[{Alon and Naor(2006)}]{Alon2006}
\bibinfo{author}{Alon, N.}, \bibinfo{author}{Naor, A.}, \bibinfo{year}{2006}.
\newblock \bibinfo{title}{{Approximating the Cut-Norm via Grothendieck's
  Inequality}}.
\newblock \bibinfo{journal}{SIAM Journal on Computing} \bibinfo{volume}{35},
  \bibinfo{pages}{787--803}.
\bibitem[{Amb\"{u}hl et~al.(2011)Amb\"{u}hl, Mastrolilli and
  Svensson}]{Ambuhl2011}
\bibinfo{author}{Amb\"{u}hl, C.}, \bibinfo{author}{Mastrolilli, M.},
  \bibinfo{author}{Svensson, O.}, \bibinfo{year}{2011}.
\newblock \bibinfo{title}{{Inapproximability Results for Maximum Edge Biclique,
  Minimum Linear Arrangement, and Sparsest Cut}}.
\newblock \bibinfo{journal}{SIAM Journal on Computing} \bibinfo{volume}{40},
  \bibinfo{pages}{567--596}.
\bibitem[{Billionnet(2004)}]{Billionnet2004}
\bibinfo{author}{Billionnet, A.}, \bibinfo{year}{2004}.
\newblock \bibinfo{title}{{Quadratic 0-1 bibliography}} \URLprefix
  \url{http://cedric.cnam.fr/fichiers/RC611.pdf}.
\bibitem[{Chang et~al.(2012)Chang, Vakati, Krause and Eulenstein}]{Chang2012}
\bibinfo{author}{Chang, W.C.}, \bibinfo{author}{Vakati, S.},
  \bibinfo{author}{Krause, R.}, \bibinfo{author}{Eulenstein, O.},
  \bibinfo{year}{2012}.
\newblock \bibinfo{title}{{Exploring biological interaction networks with
  tailored weighted quasi-bicliques.}}
\newblock \bibinfo{journal}{BMC bioinformatics} \bibinfo{volume}{13 Suppl 1},
  \bibinfo{pages}{S16}.
\bibitem[{Gillis and Glineur(2011)}]{Gillis2011}
\bibinfo{author}{Gillis, N.}, \bibinfo{author}{Glineur, F.},
  \bibinfo{year}{2011}.
\newblock \bibinfo{title}{{Low-Rank Matrix Approximation with Weights or
  Missing Data Is NP-Hard}}.
\newblock \bibinfo{journal}{SIAM Journal on Matrix Analysis and Applications}
  \bibinfo{volume}{32}, \bibinfo{pages}{1149--1165}.
\bibitem[{Hansen and Mladenovi\'{c}(2003)}]{Hansen2003}
\bibinfo{author}{Hansen, P.}, \bibinfo{author}{Mladenovi\'{c}, N.},
  \bibinfo{year}{2003}.
\newblock \bibinfo{title}{{Variable Neighborhood Search}}, in:
  \bibinfo{editor}{Glover, F.}, \bibinfo{editor}{Kochenberger, G.} (Eds.),
  \bibinfo{booktitle}{Handbook of Metaheuristics}. \bibinfo{publisher}{Kluwer
  Academic Publishers}. volume~\bibinfo{volume}{57}.
  chapter~\bibinfo{chapter}{6}, pp. \bibinfo{pages}{145--184}.
\bibitem[{Karapetyan et~al.(2009)Karapetyan, Gutin and
  Goldengorin}]{Karapetyan2009}
\bibinfo{author}{Karapetyan, D.}, \bibinfo{author}{Gutin, G.},
  \bibinfo{author}{Goldengorin, B.}, \bibinfo{year}{2009}.
\newblock \bibinfo{title}{{Empirical evaluation of construction heuristics for
  the multidimensional assignment problem}}, in: \bibinfo{editor}{Chan, J.},
  \bibinfo{editor}{Daykin, J.W.}, \bibinfo{editor}{Rahman, M.S.} (Eds.),
  \bibinfo{booktitle}{London Algorithmics 2008: Theory and Practice},
  \bibinfo{publisher}{College Publications}. pp. \bibinfo{pages}{107--122}.
\bibitem[{Koyut\"{u}rk et~al.(2005)Koyut\"{u}rk, Grama and
  Ramakrishnan}]{Koyuturk2005}
\bibinfo{author}{Koyut\"{u}rk, M.}, \bibinfo{author}{Grama, A.},
  \bibinfo{author}{Ramakrishnan, N.}, \bibinfo{year}{2005}.
\newblock \bibinfo{title}{{Compression, clustering, and pattern discovery in
  very high-dimensional discrete-attribute data sets}}.
\newblock \bibinfo{journal}{IEEE Transactions on Knowledge and Data
  Engineering} \bibinfo{volume}{17}, \bibinfo{pages}{447--461}.
\bibitem[{Koyut\"{u}rk et~al.(2006)Koyut\"{u}rk, Grama and
  Ramakrishnan}]{Koyuturk2006}
\bibinfo{author}{Koyut\"{u}rk, M.}, \bibinfo{author}{Grama, A.},
  \bibinfo{author}{Ramakrishnan, N.}, \bibinfo{year}{2006}.
\newblock \bibinfo{title}{{Nonorthogonal decomposition of binary matrices for
  bounded-error data compression and analysis}}.
\newblock \bibinfo{journal}{ACM Transactions on Mathematical Software}
  \bibinfo{volume}{32}, \bibinfo{pages}{33--69}.
\bibitem[{Lu et~al.(2011)Lu, Vaidya, Atluri, Shin and Jiang}]{Lu2011}
\bibinfo{author}{Lu, H.}, \bibinfo{author}{Vaidya, J.},
  \bibinfo{author}{Atluri, V.}, \bibinfo{author}{Shin, H.},
  \bibinfo{author}{Jiang, L.}, \bibinfo{year}{2011}.
\newblock \bibinfo{title}{{Weighted Rank-One Binary Matrix Factorization}}, in:
  \bibinfo{booktitle}{Proceedings of the Eleventh SIAM International Conference
  on Data Mining}, \bibinfo{publisher}{SIAM / Omnipress}. pp.
  \bibinfo{pages}{283--294}.
\bibitem[{L\"{u} et~al.(2010)L\"{u}, Glover and Hao}]{Lu2010}
\bibinfo{author}{L\"{u}, Z.}, \bibinfo{author}{Glover, F.},
  \bibinfo{author}{Hao, J.K.}, \bibinfo{year}{2010}.
\newblock \bibinfo{title}{{A hybrid metaheuristic approach to solving the UBQP
  problem}}.
\newblock \bibinfo{journal}{European Journal of Operational Research}
  \bibinfo{volume}{207}, \bibinfo{pages}{1254--1262}.
\newblock \URLprefix
  \url{http://linkinghub.elsevier.com/retrieve/pii/S0377221710004789},
  \DOIprefix\doi{10.1016/j.ejor.2010.06.039}.
\bibitem[{Mart\'{\i}(2003)}]{Marti2003}
\bibinfo{author}{Mart\'{\i}, R.}, \bibinfo{year}{2003}.
\newblock \bibinfo{title}{{Multi-Start Methods}}, in: \bibinfo{editor}{Glover,
  F.}, \bibinfo{editor}{Kochenberger, G.} (Eds.), \bibinfo{booktitle}{Handbook
  of Metaheuristics}. \bibinfo{publisher}{Kluwer Academic Publishers}.
  volume~\bibinfo{volume}{57}. chapter~\bibinfo{chapter}{12}, pp.
  \bibinfo{pages}{355--368}.
\bibitem[{Punnen et~al.(2012)Punnen, Sripratak and Karapetyan}]{Punnen2012}
\bibinfo{author}{Punnen, A.P.}, \bibinfo{author}{Sripratak, P.},
  \bibinfo{author}{Karapetyan, D.}, \bibinfo{year}{2012}.
\newblock \bibinfo{title}{{The bipartite unconstrained 0-1 quadratic
  programming problem: polynomially solvable cases}}.
\newblock \bibinfo{journal}{Submitted} .
\bibitem[{Shen et~al.(2009)Shen, Ji and Ye}]{Shen2009}
\bibinfo{author}{Shen, B.h.}, \bibinfo{author}{Ji, S.}, \bibinfo{author}{Ye,
  J.}, \bibinfo{year}{2009}.
\newblock \bibinfo{title}{{Mining discrete patterns via binary matrix
  factorization}}, in: \bibinfo{booktitle}{Proceedings of the 15th ACM SIGKDD
  international conference on Knowledge discovery and data mining},
  \bibinfo{publisher}{ACM Press}. pp. \bibinfo{pages}{757--766}.
\bibitem[{Tan(2008)}]{Tan2008}
\bibinfo{author}{Tan, J.}, \bibinfo{year}{2008}.
\newblock \bibinfo{title}{{Inapproximability of Maximum Weighted Edge Biclique
  and Its Applications}}, in: \bibinfo{booktitle}{Proceedings of the 5th
  international conference on Theory and applications of models of
  computation}, \bibinfo{publisher}{Springer-Verlag}. pp.
  \bibinfo{pages}{282--293}.
\bibitem[{Tanay et~al.(2002)Tanay, Sharan and Shamir}]{Tanay2002}
\bibinfo{author}{Tanay, A.}, \bibinfo{author}{Sharan, R.},
  \bibinfo{author}{Shamir, R.}, \bibinfo{year}{2002}.
\newblock \bibinfo{title}{{Discovering statistically significant biclusters in
  gene expression data}}.
\newblock \bibinfo{journal}{Bioinformatics} \bibinfo{volume}{18},
  \bibinfo{pages}{S136--S144}.
\bibitem[{Wang et~al.(2012)Wang, L\"{u}, Glover and Hao}]{Wang2012}
\bibinfo{author}{Wang, Y.}, \bibinfo{author}{L\"{u}, Z.},
  \bibinfo{author}{Glover, F.}, \bibinfo{author}{Hao, J.K.},
  \bibinfo{year}{2012}.
\newblock \bibinfo{title}{{Path relinking for unconstrained binary quadratic
  programming}}.
\newblock \bibinfo{journal}{European Journal of Operational Research}
  \bibinfo{volume}{223}, \bibinfo{pages}{595--604}.
\newblock \URLprefix
  \url{http://linkinghub.elsevier.com/retrieve/pii/S0377221712005334},
  \DOIprefix\doi{10.1016/j.ejor.2012.07.012}.

\end{thebibliography}

\end{document}